%% file: arxiv.tex
\documentclass[sigconf]{acmart}
\usepackage{graphicx}

\usepackage[inline, shortlabels]{enumitem}
\usepackage[newcommands]{ragged2e}
\usepackage{mathtools}
\usepackage{dsfont}
\usepackage{subcaption}
\usepackage{tikz}
\usetikzlibrary{patterns, decorations.pathreplacing}

\makeatletter
\newtheorem*{rep@theorem}{\rep@title}
\newcommand{\newreptheorem}[2]{%
\newenvironment{rep#1}[1]{%
  \def\rep@title{#2 \ref{##1}}%
  \begin{rep@theorem}}%
  {\end{rep@theorem}}}
\makeatother

\newreptheorem{theorem}{Theorem}
\newreptheorem{lemma}{Lemma}

\numberwithin{equation}{section}
\numberwithin{figure}{section}
\numberwithin{table}{section}

\setlist{itemsep=\smallskipamount}

\setcitestyle{numbers, sort&compress, square}

\newcommand{\captiondetail}{%
  \justifying\small\mbox{}\\\noindent}

\newcommand{\noqed}{\renewcommand{\qed}{}}

\makeatletter
\mathtoolsset{mathic}
\newtoggle{noendsub}
\toggletrue{noendsub}

\let\@oldmath\(
\renewcommand{\(}{
  \@oldmath}
\let\@oldendmath\)
\renewcommand{\)}{%
  \ifdefstring{\f@shape}{it}{%
    \iftoggle{noendsub}{\sb{}}{}\kern-\scriptspace}{}\@oldendmath}
\catcode`\$=13 
\def$#1${\(#1\)}

\let\@oldtikzpicture\tikzpicture
\def\tikzpicture{%
  \catcode`\$=3 
  \@oldtikzpicture}
\makeatother

\acmConference{arXiv.org}{May 2018}{Pittsburgh, PA, USA}
\renewcommand{\footnotetextcopyrightpermission}[1]{}

\renewcommand{\bar}[1]{%
  \mkern 2.2mu\overline{\mkern-2.2mu#1\mkern-0.8mu}\mkern 0.8mu}

\title{SRPT for Multiserver Systems}

\author{Isaac Grosof}
\affiliation{%
  \institution{Carnegie Mellon University}
  \department{Computer Science Department}
  \streetaddress{5000 Forbes Ave}
  \city{Pittsburgh}
  \state{PA}
  \postcode{15213}
  \country{USA}}
\email{igrosof@cs.cmu.edu}

\author{Ziv Scully}
\affiliation{%
  \institution{Carnegie Mellon University}
  \department{Computer Science Department}
  \streetaddress{5000 Forbes Ave}
  \city{Pittsburgh}
  \state{PA}
  \postcode{15213}
  \country{USA}}
\email{zscully@cs.cmu.edu}

\author{Mor Harchol-Balter}
\affiliation{%
  \institution{Carnegie Mellon University}
  \department{Computer Science Department}
  \streetaddress{5000 Forbes Ave}
  \city{Pittsburgh}
  \state{PA}
  \postcode{15213}
  \country{USA}}
\email{harchol@cs.cmu.edu}

\newcommand{\ldelimP}{\{}
\newcommand{\rdelimP}{\}}
\newcommand{\ldelimE}{[}
\newcommand{\rdelimE}{]}
\newcommand{\ldelimVar}{(}
\newcommand{\rdelimVar}{)}
\undef\P
\DeclarePairedDelimiterXPP{\P}[1]{\mathbf{P}}{\ldelimP}{\rdelimP}{}{%
  
  #1}
\DeclarePairedDelimiterXPP{\E}[1]{\mathbf{E}}{\ldelimE}{\rdelimE}{}{%
  
  #1}
\DeclarePairedDelimiterXPP{\Var}[1]{\mathbf{Var}}{\ldelimVar}{\rdelimVar}{}{%
  
  #1}

\def\k/{-$k$}
\def\1/{-$1$}
\def\mgk/{M/G/$k$}
\def\mg1/{M/G/$1$}
\def\mmk/{M/M/$k$}
\newcommand{\indicator}{\mathds{1}}
\newcommand{\A}{{(1)}}
\newcommand{\B}{{(k)}}
\newcommand{\st}{\textrm{\textnormal{st}}}
\newcommand{\srpt}{\textrm{\textnormal{SRPT\1/}}}
\newcommand{\srptk}{\textrm{\textnormal{SRPT\k/}}}
\newcommand{\psjf}{\textrm{\textnormal{PSJF\1/}}}
\newcommand{\psjfk}{\textrm{\textnormal{PSJF\k/}}}
\newcommand{\rsk}{\textrm{\textnormal{RS\k/}}}
\newcommand{\rs}{\textrm{\textnormal{RS\1/}}}
\newcommand{\fb}{\textrm{\textnormal{FB\1/}}}
\newcommand{\fbk}{\textrm{\textnormal{FB\k/}}}
\newcommand{\rwork}{\mathtt{RelWork}}
\newcommand{\vwork}{\mathtt{VirtWork}}

\newcommand{\optk}{\textnormal{OPT\k/}}
\newcommand{\busy}{\mathtt{RelBusy}}
\begin{document}

\begin{abstract}
  The Shortest Remaining Processing Time (SRPT) scheduling policy
  and its variants have been extensively studied
  in both theoretical and practical settings.
  While beautiful results are known for single-server SRPT,
  much less is known for \emph{multiserver SRPT}.
  In particular, stochastic analysis of the \mgk/ under multiserver SRPT
  is entirely open.
  Intuition suggests that multiserver SRPT
  should be optimal or near-optimal for minimizing mean response time.
  However, the only known analysis of multiserver SRPT
  is in the worst-case adversarial setting,
  where SRPT can be far from optimal.
  In this paper, we give the \emph{first stochastic analysis}
  bounding mean response time of the \mgk/ under multiserver SRPT.
  Using our response time bound,
  we show that multiserver SRPT has \emph{asymptotically optimal}
  mean response time in the heavy-traffic limit.
  The key to our bounds is a strategic combination
  of stochastic and worst-case techniques.
  Beyond SRPT,
  we prove similar response time bounds and optimality results
  for several other multiserver scheduling policies.
\end{abstract}

\maketitle

\section{Introduction}

The Shortest Remaining Processing Time (SRPT) scheduling policy
and variants thereof have been deployed in many computer systems,
including web servers \cite{Harchol-Balter:2003:SSI:762483.762486},
networks \cite{mangharam2003size},
databases \cite{guirguis2009adaptive},
operating systems \cite{bunt1976scheduling}
and FPGA layout systems \cite{10.1007/11596356_50}.
SRPT has also long been a topic of fascination for queueing theorists
due to its optimality properties.
In 1966, the mean response time for SRPT was first derived \cite{schrage1966queue},
and in 1968 SRPT was shown to minimize mean response time
both in a stochastic sense and in a worst-case sense \cite{schrage1968letter}.
However, these beautiful optimality results and the analysis of SRPT
are only known for \emph{single-server} systems.
Almost nothing is known for \emph{multiserver} systems, such as the \mgk/,
even for the case of just $k=2$ servers.

The SRPT policy for the \mgk/ is defined as follows:
at all times,
the $k$ jobs with smallest remaining processing time receive service,
preempting jobs in service if necessary.

We assume a central queue, meaning any job can be
dispatched or migrated to any server at any time,
and a preempt-resume model,
meaning preemption incurs no cost or loss of work.

It seems believable that SRPT should minimize mean response time in multiserver systems
because it gives priority to the jobs which will finish soonest,
which seems like it should minimize the number of jobs in the system.
However, it was shown in 1997 that SRPT is not optimal for
multiserver systems in the worst case \cite{leonardi1997approximating, Leonardi}.
That is, one can come up with an adversarial arrival sequence for which
the mean response time under SRPT is larger that the optimal mean response time.
In fact, the ratio by which SRPT's mean response time
exceeds the optimal mean response time can be arbitrarily large
\cite{leonardi1997approximating, Leonardi}.

The fact that multiserver SRPT is not optimal in the worst case
provokes a natural question about the \emph{stochastic} case.
\begin{quote}
  \emph{Is SRPT optimal or near-optimal
    for minimizing mean response time in the the \mgk/?}
\end{quote}
Unfortunately, this question is entirely open.  Not only is it not
known whether SRPT is optimal,
but multiserver SRPT has also eluded stochastic analysis.
\begin{quote}
  \emph{What is the mean response time for the \mgk/ under SRPT?}
\end{quote}
The purpose of this paper is to answer both of these questions
in the high-load setting.
Under low load, response time is dominated by service time,
which is not affected by the scheduling policy.
In contrast, under high load, response time is dominated by queueing time,
which can vary wildly under different scheduling policies.
We thus focus on the high-load setting,
and specifically on the \emph{heavy-traffic limit} as load approaches capacity.

Our main result is that,
under mild assumptions on the service requirement distribution,
\begin{quote}
  \emph{SRPT is an optimal multiserver policy
    for minimizing mean response time in the \mgk/
    in the heavy-traffic limit.}
\end{quote}
We also give the
\emph{first mean response time bound for the \mgk/ under SRPT}.
The bound is valid for all loads and is tight for load near capacity.

In addition to SRPT,
we give the
\emph{first mean response time bounds for the \mgk/
  with three other scheduling policies},
specifically Preemptive Shortest Job First (PSJF)
\cite{wierman2003classifying},
Remaining Size Times Original Size (RS)
\cite{wierman2005nearly, Hyytia:2012:MSH:2254756.2254763},
and Foreground-Background (FB)
\cite{nuyens2008foreground}.
Our bounds imply that in the heavy-traffic limit,
under the same mild assumptions as for SRPT above,
\begin{itemize}
\item
  multiserver PSJF and RS are also optimal multiserver scheduling policies; and
\item
  multiserver FB is optimal in the same setting where single-server FB is optimal
  \cite{righter1989scheduling},
  which is when the service requirement distribution has decreasing hazard rate
  and the scheduler does not have access to job sizes.
\end{itemize}

Our approach to analyzing SRPT on $k$~servers
is to compare its performance to that of
SRPT on a single server which is $k$ times as fast,
where both systems have the same arrival rate~$\lambda$
and service requirement distribution~$S$.
Specifically,
let \emph{SRPT\k/} be the policy which uses multiserver SRPT
on $k$ servers of speed $1/k$,
as shown in Figure~\ref{fig:server_diagram}.
Ordinary SRPT on a single server is simply SRPT\1/.
The \emph{system load} $\rho = \lambda\E{S}$ is
the average rate at which work enters the system.
The maximal total rate at which the $k$ servers can do work is~$1$,
so the system is stable for $\rho < 1$, which we assume throughout.

\begin{figure}
  \centering
  \input{fig_server_diagram}
  \caption{Single-server and $k$-server systems}
  \label{fig:server_diagram}
\end{figure}
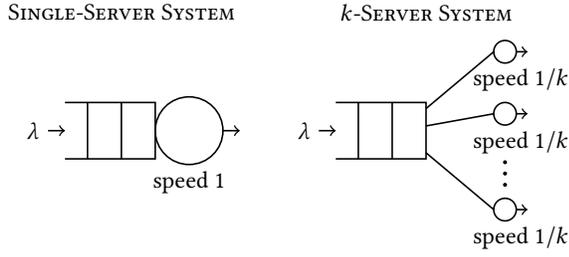

Our main result is that in the $\rho \to 1$ limit,
the mean response time under SRPT\k/, $\E[\big]{T^\srptk}$,
approaches the mean response time under SRPT\1/, $\E[\big]{T^\srpt}$.
Because SRPT\1/ minimizes response time among all scheduling policies,
this means that SRPT\k/ is asymptotically optimal among $k$-server policies.
In particular, let OPT\k/ be the optimal $k$-server policy. Then
\begin{equation*}
  \E[\big]{T^\srpt} \le \E[\big]{T^\optk} \le \E[\big]{T^\srptk},
\end{equation*}
so showing that $\E[\big]{T^\srptk}$ approaches $\E[\big]{T^\srpt}$ as $\rho \to 1$
also shows that $\E[\big]{T^\srptk}$ approaches $\E[\big]{T^\optk}$ as $\rho \to 1$.

Specifically, we prove the following sequence of theorems.

Our first theorem is an upper bound on
the mean response time of a job of size~$x$ under SRPT\k/,
written $\E[\big]{T^\srptk(x)}$.
As in the classic SRPT\1/ analysis \cite{schrage1966queue},
the response time of a job of size~$x$ depends on
the system load contributed by jobs of size at most~$x$,
written~$\rho_{\le x}$ (see Definition~\ref{def:busy_period}).

\begin{reptheorem}{thm:improved-response}
  In an \mgk/,
  the mean response time of a job of size~$x$ under SRPT\k/ is bounded by
  \[\E[\big]{T^\srptk(x)} \le \frac{\int_0^x \lambda t^2f_S(t) \,dt}{2(1-\rho_{\le x})^2}
  + \frac{k\rho_{\le x}x}{1-\rho_{\le x}}
  + \int_0^x \frac{k}{1-\rho_{\le t}} \,dt,\]
  where $f_S(\cdot)$ is the probability density function of
  the service requirement distribution~$S$.
\end{reptheorem}

The bound given in Theorem~\ref{thm:improved-response}
holds for any load~$\rho$ and any service requirement distribution~$S$.
We use this bound to prove that, under mild conditions on~$S$,
the performance of SRPT\k/ approaches that of SRPT\1/
in the $\rho \to 1$ limit,
which implies asymptotic optimality of SRPT\k/.

\begin{reptheorem}{thm:ratio}
  In an \mgk/ with any service requirement distribution~$S$ which is either
  \begin{enumerate*}[(i)]
  \item
    bounded or
  \item
    unbounded with a tail function which has upper Matuszewska index\footnote{This technical condition is roughly equivalent to finite variance. See Section~\ref{sec:matuszewska} or Appendix~\ref{app:matuszewska}.} less than~$-2$,
  \end{enumerate*}
  \[\lim_{\rho \rightarrow 1} \frac{\E[\big]{T^\srptk}}{\E*{T^\srpt}} = 1.\]
\end{reptheorem}

The technique by which we bound response time under SRPT\k/
is widely generalizable.
We also use it to give mean response time bounds and optimality results
for PSJF\k/, RS\k/, and FB\k/
(see Section~\ref{sec:other-policies}).

Our approach is inspired by two very different worlds:
the stochastic world and the adversarial worst-case world.
Purely stochastic approaches are
difficult to generalize to the \mgk/ for many reasons,
including the fact that multiserver systems are not work-conserving.
Purely adversarial worst-case analysis is easier
but leads to weak bounds when directly applied to the stochastic setting.
For instance, \citet{leonardi1997approximating, Leonardi} show that for an adversarial arrival sequence,
SRPT\k/ has worse mean response time than
the optimal offline $k$-server policy
by a factor of $\Omega(\log(\min(n/k,P))$,
where $n$ is the total number of jobs in the arrival sequence
and $P$ is the ratio of the smallest and largest job sizes.
This factor can be arbitrarily large in the context of the \mgk/,
because $n \to \infty$ if the arrival sequence is an infinite Poisson process,
and $P \to \infty$ if the service requirement distribution is unbounded
or allows for arbitrarily small jobs.

What makes our analysis work is a
strategic combination of the stochastic and worst-case techniques.
We use the more powerful stochastic tools where possible
and use worst-case techniques to bound variables for which
exact stochastic analysis is intractable.

\section{Prior Work}
Countless papers have been published on the stochastic analysis of the SRPT policy
in the single-server model over the last 52 years,
beginning in 1966 with \citeauthor{schrage1966queue}'s response time analysis
of the \mg1/ queue under SRPT \cite{schrage1966queue},
which was followed shortly by
the proof of SRPT's optimality \cite{schrage1968letter}.
SRPT remains a major topic of study today.
There have been beautiful works on analyzing the tail of response time
\cite{borst2002heavy, borst2003impact, Boxma},
the fairness of SRPT
\cite{bansal2001analysis, wierman2003classifying}
and SRPT in different models,
such as energy-aware control \cite{10.1007/978-3-319-43425-4_7}.

However, all of these works analyze \emph{single-server} SRPT.
We give the first analysis of multiserver SRPT.
While single-server SRPT minimizes mean response time, multiserver SRPT does not\footnote{%
  It has been claimed that multiserver SRPT is optimal
  under the additional assumption that all servers are busy at all times
  \cite[Theorem~2.1]{down2006multi}.
  However, the proof has an error.
  See Appendix~\ref{app:interchange}.} \cite{leonardi1997approximating, Leonardi}.
We show that multiserver SRPT approaches optimality in heavy traffic.

\subsection{Single-Server SRPT in Heavy Traffic}
\label{sec:matuszewska}

While the exact mean response time analysis of single-server SRPT is known, it is in the form of a triply nested integral. Therefore, it is useful to have a simpler formula for mean response time. Many papers have derived such a formula under heavy traffic \cite{Lin, bansal2005average, bansal2006handling, down2009fluid}.

Heavy traffic analysis describes the behavior of a queueing system in the limit as load approaches capacity. The most general heavy-traffic analysis of the mean response time of single-server SRPT is due to \citet{Lin}, who characterize the asymptotic behavior of mean response time for general service requirement distributions.
They consider three categories of service requirement distributions and give an asymptotic analysis of the mean response time of each:
\begin{itemize}
\item bounded distributions,
\item distributions whose tail has upper Matuszewska index\footnote{%
    See Appendix~\ref{app:matuszewska}.} less than~$-2$, and
\item distributions whose tail has lower Matuszewska index greater than~$-2$.
\end{itemize}
The first and second categories above roughly correspond to the distribution having finite variance, while the third roughly corresponds to the distribution having infinite variance.

In this paper, we restrict our heavy-traffic results to the first two categories, focusing on service requirement distributions that are either bounded or whose tails have upper Matuszewska index less than $-2$.
We build on the work of \citet{Lin} to give the first heavy-traffic analysis of multiserver SRPT. In particular, we demonstrate that in the heavy-traffic limit, the mean response time of SRPT in a multiserver system with $k$ servers approaches that of SRPT in a single-server system which runs $k$ times faster (see Figure~\ref{fig:server_diagram}).

\subsection{The Multiserver Priority Queue}

While there is no existing stochastic analysis of multiserver SRPT,
there is some analysis of multiserver priority queues.
In a multiserver priority queue,
it is assumed that there are finitely many classes of jobs (typically two)
with exponential or phase-type service requirement distributions.
Thus, the system can be modeled as a multidimensional Markov chain.
\Citet{Mitrani} give an exact analysis of the two class multiserver system with preemptive priority between the job classes and exponential service times within each class.
\Citet{Sleptchenko} extend this analysis to hyperexponential service requirement distributions,
and \citet{Osogami} extend it further still to support phase-type service requirement distributions and any constant number of preemptive priority classes.
However, the solutions found through these extensions can take a very long time to calculate, requiring more time with every added server, priority class, or state in the phase-type distribution.

Our analysis goes beyond the multiclass setting by handling an arbitrary service requirement distribution and a policy, namely SRPT\k/, with an infinite set of priorities.
Furthermore, our analysis produces a \emph{closed-form} result,
in contrast to the numerical results of these prior works.

\subsection{Multiserver SRPT in the Worst Case}

While stochastic analysis of mutliserver SRPT is open,
multiserver SRPT has been well studied in the worst-case setting.
Worst-case analysis considers an \emph{adversarially chosen} sequence of job arrival times and service requirements.
An online policy (which does not know the arrival sequence) such as SRPT\k/ is typically compared to the optimal offline policy (which knows the arrival sequence). In the worst-case setting, a policy is a $c$-approximation if its mean response time is at most $c$ times the mean response time of the offline optimal policy on any arrival sequence.

\citet{leonardi1997approximating, Leonardi} analyze SRPT\k/ in the worst-case setting
under the assumptions that
\begin{enumerate*}[(1)]
\item
  there are $n$ jobs in the arrival sequence and
\item
  the ratio of the largest and smallest service requirements in the arrival sequence is~$P$.
\end{enumerate*}
They show that SRPT\k/ is an $O(\log(\min(n/k, P)))$-approximation for mean response time, where $n$ is the total number of jobs.
They also show that any online policy is at least an $\Omega(\log(\min(n/k, P)))$-approximation.
This shows that no online policy has a better approximation ratio than SRPT\k/ by more than a constant factor.

Unfortunately, directly applying the $O(\log(\min(n/k, P)))$ bound on SRPT\k/ to the \mgk/ is not helpful for two reasons.
First, the arrival process is an infinite Poisson process, so $n \to \infty$.
Second, often the maximum job size is unbounded or the minimum job size is arbitrarily small, so $P \to \infty$ as well.

SRPT has also been considered in other multiserver models.
For example, \citet{Avrahami:2003:MTF:777412.777415} analyze the immediate dispatch setting,
in which each server has a queue and jobs are dispatched to these queues on arrival.
Each server can only serve the jobs in its queue, and jobs cannot migrate between queues.
Within each queue, jobs are served according to SRPT.
\Citet{Avrahami:2003:MTF:777412.777415} give a dispatch policy called IMD
which achieves the same $O(\log(\min(n/k, P)))$-approximation as SRPT\k/,
even when compared to the optimal offline policy with migrations.
Again, directly applying this to the \mgk/ is problematic because $n \to \infty$ and $P \to \infty$.

In contrast with these worst-case results, we show that in the stochastic setting, SRPT\k/ is asymptotically optimal policy for mean response time in the heavy-traffic limit. Our result holds for an extremely general class of service requirement distributions, including distributions which are unbounded and/or have arbitrarily small jobs.

\subsection{Other Prior Work}

\Citet{Gong} propose a \emph{single-server} policy called K-SRPT which is superficially similar to our SRPT\k/. Specifically, K-SRPT shares the processor between the $k$ jobs in the system with least remaining time. That is, K-SRPT is a hybrid of processor sharing (PS) and SRPT. Crucially, when fewer than $k$ jobs are in the system, K-SRPT allows each job to receive an increased share of the maximum service rate, ensuring work conservation.
In contrast, our SRPT\k/ model never allows a job to receive more than $1/k$ of the maximum service rate of the system, since a job cannot run on more than one server at once. This means SRPT\k/ is not work-conserving, which makes it difficult to analyze.

\section{Model}
We study scheduling policies for the \mgk/ queue. We write $\lambda$ for the arrival rate, $S$~for the service requirement distribution, and $k$ for the number of servers.
The rate at which any given server completes work is~$1/k$.
That is, a job with a service requirement, or \emph{size},
of $x$ needs to be served for time~$kx$ to complete.
The $k$ servers all together have total service rate~$1$.

The \emph{load} of the \mgk/ system, namely the average rate at which work arrives, is
\begin{equation*}
  \rho = \lambda\E{S}.
\end{equation*}
That is, jobs arrive at rate~$\lambda$ jobs per second,
each contributing $\E{S}$ work in expectation.
We can view $\E{S} = 1/(k\mu)$,
where $1/\mu$ is the expected amount of time a job needs to be served to complete.
We assume a stable system, meaning $\rho < 1$,
and a preempt-resume model, meaning that preemption incurs no cost or loss of work.

We will analyze systems in the \emph{heavy-traffic limit},
which is the limit as $\rho \rightarrow 1$.
More precisely, this is the limit as $\lambda \rightarrow 1/\E{S}$ for fixed~$S$.

We analyze and compare systems with $k=1$ and general $k$.
An example of each is shown in Figure~\ref{fig:server_diagram}.
Note that in our model,
the \mg1/ and \mgk/ systems have the same load~$\rho$.

The primary policy we study is the SRPT\k/ policy, which is the Shortest Remaining Processing Time policy on $k$ servers. At every moment in time, SRPT\k/ serves the $k$ jobs with smallest remaining processing time. If there are fewer than $k$ jobs in the system, every job receives service, which leaves some servers idle.
Note SRPT\1/ is the usual single-server SRPT policy.

\section{Background and Challenges}

Our approach to analyzing response time under SRPT\k/
is to compare it with SRPT\1/.
As such, we begin this section by briefly reviewing the analysis of SRPT\1/,
specifically focusing on the definitions and formulations
that will come up in the SRPT\k/ analysis.
We then outline why the SRPT\1/ analysis
does not easily generalize to SRPT\k/ with $k > 1$ servers.

\subsection{SRPT-1 Tagged Job Tutorial}
\label{sub:srpt-1}

We now review the technique used by \citet{schrage1966queue} to analyze SRPT\1/.
Consider a particular ``tagged'' job~$j$, of size $x$,
arriving to a \emph{random system state}
drawn from the system's steady-state distribution.
We denote $j$'s response time by $T^\srpt(x)$.
Of course, $T^\srpt(x)$ is a random variable
which depends on both the random arrivals that occur after~$j$
and the random queue state that $j$ observes upon its own arrival.

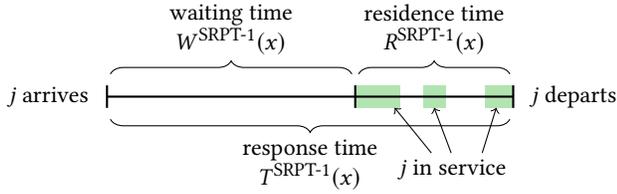
\begin{figure}
  \centering
  \input{fig_waiting_residence}
  \caption{Response time of the tagged job~$j$ of size~$x$
    is the sum of waiting time and residence time}
  \label{fig:waiting-residence}
\end{figure}

We split the analysis of $T^\srpt(x)$ into two parts, shown in Figure~\ref{fig:waiting-residence}:
\begin{itemize}
\item
  \emph{waiting time}~$W^\srpt(x)$,
  the time between $j$'s arrival and the moment $j$ first enters service;
  and
\item
  \emph{residence time}~$R^\srpt(x)$,
  the time between the moment $j$ first enters service and $j$'s departure.
\end{itemize}
Given waiting time and residence time,
response time is simply
\begin{equation*}
  T^\srpt(x) = W^\srpt(x) + R^\srpt(x).
\end{equation*}

Under SRPT\1/, $j$ has priority over all jobs with larger remaining size than itself,
so such jobs do not impact $j$'s response time.

\begin{definition}
  \label{def:relevant}
  Suppose job~$j$ has remaining size~$x$.
  A job~$\ell$ is \emph{relevant} to job~$j$
  if $\ell$ has remaining size at most~$x$.
  Otherwise $\ell$ is \emph{irrelevant} to~$j$.
\end{definition}

In particular, we will often consider which jobs are relevant to the tagged job $j$. We will simply call jobs ``relevant'' and ``irrelevant'' when the comparison is clear from context.
For the purpose of analyzing $j$'s response time,
we can ignore all jobs which are irrelevant to $j$.

During $j$'s waiting time, the server is only doing \emph{relevant work},
namely work that is due to a relevant job.
The total amount of work done is the sum of
\begin{itemize}
\item
  relevant work due to relevant jobs that were in the system when $j$ arrived
  and
\item
  relevant work due to relevant jobs that arrived after~$j$.
\end{itemize}
To analyze $j$'s waiting time, we make use of a concept called a ``busy period''.

\begin{definition}
  A \emph{busy period} started by (possibly random) amount of work~$V$,
  written~$B(V)$,
  is the amount of time it takes for a work-conserving system that starts with $V$ work
  to become empty.
\end{definition}

Busy periods are very useful because their length depends
only on the initial amount of work and the arrival process,
not on the service policy or the number of jobs in the system.

In the SRPT\1/ system, we do not have to wait for the system to become completely empty
for $j$ to start receiving service.
We only have to wait for the system to become empty of relevant work.
We capture this with the concept of a ``relevant busy period''.

\begin{definition}
  \label{def:busy_period}
  A \emph{relevant busy period} for a job of size~$x$ started by
  (possibly random) amount of work~$V$,
  written~$B_{\le x}(V)$,
  is the amount of time it takes for
  a work-conserving system that starts with $V$ work to become empty,
  where only arrivals of size at most~$x$, the relevant arrivals,
  are admitted to the system.
  A relevant busy period has expectation
  \begin{equation*}
    \E{B_{\le x}(V)} = \frac{\E{V}}{1 - \rho_{\le x}}.
  \end{equation*}
  Above, $\rho_{\le x}$ is the \emph{relevant load} for a job of size~$x$,
  which is the total load due to relevant jobs.
  Its value is
  \begin{equation*}
    \rho_{\le x} = \lambda\E{S\indicator(S \le x)},
  \end{equation*}
  where $\indicator(\cdot)$ is the indicator function.
\end{definition}

This means $j$'s waiting time is a relevant busy period
started by the amount of relevant work that the tagged job $j$ sees on arrival.
By the PASTA property (Poisson Arrivals See Time Averages)
\cite{pasta_wolff},
the distribution for the amount of relevant work $j$ sees
is the steady-state distribution.

\begin{definition}
  \label{def:rwork}
  The \emph{steady-state relevant work} for a job of size~$x$ under SRPT\1/,
  written $\rwork^\srpt_{\le x}$,
  is the sum of remaining sizes of all jobs with remaining size at most~$x$
  observed at a random point in time.
  (An analogous definition applies to SRPT\k/.)
\end{definition}

By the above discussion, $j$'s waiting time is
\begin{equation*}
  W^\srpt(x) = B_{\le x}\bigl(\rwork^\srpt_{\le x}\bigr).
\end{equation*}
The analysis of $\rwork^\srpt_{\le x}$ is known \cite{schrage1966queue}
but outside the scope of this tutorial.

The residence time of~$j$ can be analyzed in a similar way.
At the start of $j$'s residence time,
the SRPT\1/ policy serves $j$, so $j$, which has remaining~$x$,
must be the job with the smallest remaining size in the system.
This means the system is effectively empty from $j$'s perspective,
because all work relevant to $j$ is gone.

The only work that will be done from this point until $j$ is completes is
work on $j$ itself and relevant arrivals.
Because $j$'s residence time starts with its own work~$x$
and ends when that work is done,
we can stochastically upper bound $j$'s residence time
as a relevant busy period:
\begin{equation*}
  R^\srpt(x) \le_\st B_{\le x}(x).
\end{equation*}
The reason this bound is not tight is because
$j$'s remaining size decreases during service,
which changes the cutoff for relevant jobs.
An exact analysis of $R^\srpt(x)$ is known \cite{schrage1966queue}
but outside the scope of this tutorial.

\subsection{Why the Tagged Job Analysis is Hard for SRPT\k/}
\label{sec:srpt-k-hard}
Having summarized the analysis of SRPT\1/, it is natural to ask:
why does a similar strategy not work for SRPT\k/?
The primary difficulty is that multiserver systems are
\emph{not work-conserving},
which manifests in two ways.

First, analyzing \emph{busy periods} relies on work conservation,
namely the fact that the server is doing work at rate~$1$ whenever the system is not empty.
This allows for many simplifications.
For instance, in Definition~\ref{def:busy_period},
we define busy periods as being started as a total amount of work,
without worrying exactly how that work is divided among jobs.
In a $k$-server system, work is only done at rate~$1$ if there are $k$ or more jobs in the system.
Thus, the exact rate at which work is done varies over time
depending on the number of jobs in the system, making it difficult to analyze.

Second, analyzing the \emph{steady-state relevant work}
relies on work conservation.
The analysis of $\rwork^\srpt_{\le x}$ by \citet{schrage1966queue}
relies on being able to equate $\rwork^\srpt_{\le x}$ to the total work
in a simpler first-come-first-served system.
Equality of remaining work only holds if both systems are work-conserving.
The fact that SRPT\k/ is not work-conserving means that we can't make such an argument.

\section{Analysis of SRPT\k/}
\label{sec:analysis}

As explained in Section~\ref{sec:srpt-k-hard}, traditional tagged job analysis
cannot be applied to SRPT\k/ because SRPT\k/ is not work-conserving.
Our approach is to find a way to make SRPT\k/ appear work-conserving
while the tagged job~$j$ is in the system.
We do this by introducing the
new concept of \emph{virtual work}.
Virtual work encapsulates all of the time that the servers spend either idle
or working on irrelevant jobs while $j$ is in the system.
By thinking of these times as ``virtual work'',
the system appears to be work-conserving while $j$ is in the system,
allowing us to bound the response time of~$j$.

Consider a tagged job~$j$ of size~$x$.
Recall from Definition~\ref{def:relevant}
that only jobs of remaining size at most~$x$
are \emph{relevant} to $j$ when $j$ arrives.
We will bound $j$'s response time by bounding
the \emph{total amount of server activity} between $j$'s arrival and departure.
Between $j$'s arrival and departure,
each server can be doing one of four categories of work.
\begin{itemize}
\item
  \emph{Tagged work:}
  serving~$j$.
\item
  \emph{Old work:}
  serving a job which is relevant to $j$ that was in the system upon~$j$'s arrival.
\item
  \emph{New work:}
  serving a job which is relevant to $j$ that arrived after~$j$.
\item
  \emph{Virtual work:}
  either idling or serving an job which is irrelevant to $j$.
\end{itemize}
The response time of~$j$ is exactly the total of
tagged, old, new, and virtual work.
The main idea behind our analysis is to bound this total
by a single (work-conserving) relevant busy period
(see Definition~\ref{def:busy_period}).

We already know a few facts about the four categories of work.
\begin{itemize}
\item
  Tagged work is $j$'s size~$x$.
\item
  Old work is equal to
  the amount of relevant work seen by $j$ upon arrival.\footnote{%
    One might worry that an \emph{old job} that is irrelevant when $j$ arrives
    could later become relevant to~$j$,
    and therefore be part of old work,
    but this does not occur under SRPT\k/.}
  By the PASTA property \cite{pasta_wolff},
  this is $\rwork^\srptk_{\le x}$,
  the steady state amount of relevant work for a job of size~$x$
  (see Definition~\ref{def:rwork}).
\item
  New work is bounded by all jobs which are relevant to a job of remaining size $x$
  that arrive during a
  relevant busy period $B_{\le x}(\cdot)$ started by
  tagged, old, and virtual work.\footnote{%
    One might worry that a \emph{new job} that is irrelevant when it arrives
    could later become relevant to~$j$,
    and therefore be part of new work,
    but this does not occur under SRPT\k/.}
  This is only an upper bound because we ignore the fact that
  $j$'s remaining size decreases as $j$ is served,
  which changes the size cutoff for relevant jobs.
\item
  Virtual work is as of yet unknown.
  We denote with the random variable $\vwork^\srptk(x)$
  the amount of virtual work done while $j$ is in the system.
\end{itemize}
Taken together, these yield the bound
\begin{equation}
  \label{eq:busy_bound}
  T^\srptk(x)
  \le_\st
  B_{\le x}\Bigl(x + \rwork^\srptk_{\le x} + \vwork^\srptk(x)\Bigr).
\end{equation}
Our task in the remainder of this section
is to bound $\rwork^\srptk_{\le x}$ and $\vwork^\srptk(x)$
as tightly as we can.
We use worst-case methods to bound $\vwork^\srptk(x)$
and a combination of stochastic
and worst-case methods to bound $\rwork^\srptk_{\le x}$.

\subsection{Virtual Work}

We start by bounding $\vwork^\srptk(x)$,
the virtual work done while $j$ is in the system.
A purely stochastic analysis of virtual work would be very difficult.
Fortunately, a simple worst-case bound suffices for our purposes.
The key is that a server can do virtual work
\emph{only while $j$ is in service at a different server}.
This is because SRPT\k/ never allows an irrelevant job
to have priority over~$j$.

\begin{lemma}
  \label{lem:vwork}
  The virtual work is bounded by
  \begin{equation*}
    \vwork^\srptk(x) \le (k - 1)x.
  \end{equation*}
\end{lemma}

\begin{proof}
  Virtual work only occurs while $j$ is in service.
  The maximum possible virtual work is achieved by
  all $k - 1$ other servers doing virtual work whenever $j$ is in service.
  Each server does work at rate~$1/k$.
  This means $j$ is in service for time~$kx$,
  during which virtual work is done at rate at most $(k - 1)/k$.
\end{proof}

\subsection{Relevant Work}
\label{sec:srpt-rwork}

Our next task is to bound $\rwork^\srptk_{\le x}$,
the steady state amount of relevant work for a job of size~$x$ under SRPT\k/.
As with virtual work,
a purely stochastic analysis of relevant work would be very difficult.
We therefore take the following hybrid approach.
We consider a pair of systems,
one using SRPT\1/ and the other using SRPT\k/,
experiencing the same arrival sequence.
We compare the amounts of relevant work in each system,
giving a \emph{worst-case bound for the difference}.
This allows us to use
the previously known \emph{stochastic analysis} of $\rwork^\srpt_{\le x}$
to give a stochastic bound for $\rwork^\srptk_{\le x}$.

Consider running a pair of systems under the same job arrival sequence:
\begin{itemize}
\item
  \emph{System~$1$}, which schedules using SRPT\1/; and
\item
  \emph{System~$k$}, which schedules using SRPT\k/.
\end{itemize}
For any time~$t$,
let $\rwork^\A_{\le x}(t)$ be the amount of relevant work in System~$1$ at~$t$,
and similarly for $\rwork^\B_{\le x}(t)$.
Our goal is to give a worst-case bound for
the difference in relevant work between Systems~$1$ and~$k$,
\begin{equation*}
  \Delta_{\le x}(t) = \rwork^\B_{\le x}(t) - \rwork^\A_{\le x}(t).
\end{equation*}

To bound $\Delta_{\le x}(t)$, we split times~$t$ into two types of intervals:
\begin{itemize}
\item
  \emph{few-jobs intervals}, during which there are
  fewer than~$k$ relevant jobs at a time in System~$k$;
  and
\item
  \emph{many-jobs intervals}, during which there are
  at least~$k$ relevant jobs at a time in System~$k$.
\end{itemize}
A similar type of splitting was used by
\citet{leonardi1997approximating, Leonardi}.

As a reminder, a job is \emph{relevant} if its remaining size is at most~$x$
and \emph{irrelevant} otherwise
(see Definition~\ref{def:relevant}).
Note that many-jobs intervals are defined only in terms of System~$k$,
so System~$1$ may or may not have relevant jobs during a many-jobs interval.

\begin{lemma}
  \label{lem:delta}
  For any arrival sequence and at any time~$t$,
  the difference between the relevant work in System~$1$
  and the relevant work in System~$k$ is bounded by
  \begin{equation*}
    \Delta_{\le x}(t) \le kx.
  \end{equation*}
\end{lemma}

\begin{proof}
  Any time~$t$ is in either a few-jobs interval or a many-jobs interval.
  The case where $t$ is in a few-jobs interval is simple:
  there are at most $k - 1$ relevant jobs in System~$k$ at time~$t$,
  each of remaining size at most~$x$, so
  \begin{equation*}
    \Delta_{\le x}(t) \le \rwork^\B_{\le x}(t) \le (k - 1)x.
  \end{equation*}

  Suppose instead that $t$ is in a many-jobs interval.
  Let time~$s$ be the start of the many-jobs interval containing~$t$.
  We will show
  \begin{equation*}
    \Delta_{\le x}(t) \le \Delta_{\le x}(s) \le kx.
  \end{equation*}

  We first show that $\Delta_{\le x}(t) \le \Delta_{\le x}(s)$.
  Let
  \begin{align*}
    D^\A &= \rwork^\A_{\le x}(t) - \rwork^\A_{\le x}(s) \\
    D^\B &= \rwork^\B_{\le x}(t) - \rwork^\B_{\le x}(s)
  \end{align*}
  be the change in relevant work from $s$ to~$t$ in Systems~$1$ and~$k$,
  respectively.
  Because
  \begin{equation*}
    \Delta_{\le x}(t) - \Delta_{\le x}(s) = D^\B - D^\A,
  \end{equation*}
  it suffices to show $D^\B \leq D^\A$.

  We can write $D^\A$ as a sum of three components,
  \begin{equation*}
    D^\A = \mathtt{Arrivals}^\A + \mathtt{NewlyRelevant}^\A - \mathtt{Served}^\A,
  \end{equation*}
  which are defined as follows.
  \begin{itemize}
  \item
    $\mathtt{Arrivals}^\A$ is the relevant work added during $[s, t]$
    due to relevant new arrivals.
  \item
    $\mathtt{NewlyRelevant}^\A$ is the relevant work added during $[s, t]$
    due to the server serving irrelevant jobs
    until they reach remaining size~$x$,
    at which point they become relevant.
    For our purposes,
    all that matters is that $\mathtt{NewlyRelevant}^\A \ge 0$.
  \item
    $\mathtt{Served}^\A$ is the
    amount of relevant work done by the server during $[s, t]$.
    System~$1$ does relevant work at rate~$1$ if it has any relevant jobs
    and rate~$0$ otherwise,
    so $\mathtt{Served}^\A \le t - s$.
  \end{itemize}
  We define analogous quantities for System~$k$
  and compare them to their System~$1$ counterparts.
  \begin{itemize}
  \item
    $\mathtt{Arrivals}^\B = \mathtt{Arrivals}^\A$
    because the two systems experience the same arrivals.
  \item
    $\mathtt{NewlyRelevant}^\B = 0$
    because $[s, t]$ is within a many-jobs interval,
    during which System~$k$ has at least $k$ relevant jobs.
    Therefore, there is never an opportunity for an irrelevant job
    to be served and become relevant.
    In particular,
    \begin{equation*}
      \mathtt{NewlyRelevant}^\B \le \mathtt{NewlyRelevant}^\A.
    \end{equation*}
  \item
    $\mathtt{Served}^\B = t - s$
    because $[s, t]$ is within a many-jobs interval,
    during which System~$k$ has at least $k$ relevant jobs.
    Therefore, its servers do relevant work at combined rate~$1$
    during all of $[s, t]$.
    In particular,
    \begin{equation*}
      \mathtt{Served}^\B \ge \mathtt{Served}^\A.
    \end{equation*}
  \end{itemize}
  The three comparisons above imply $D^\B \le D^\A$, as desired.

  All that remains is to show $\Delta_{\le x}(s) \le kx$.
  Recall that $s$ is the start of a many-jobs interval.
  There are two ways to enter a many-jobs interval.
  In both cases, we show that $\Delta_{\le x}(s) \leq kx$.

  One way a many-jobs interval can start is when
  \emph{a relevant job arrives while System~$k$ has $k - 1$ relevant jobs}.
  The same arrival occurs in System~$1$,
  so $\Delta_{\le x}(s) = \Delta_{\le x}(s^-)$,
  where $s^-$ is the instant before the arrival.
  But $s^-$ is the end of a few-jobs interval,
  during which System~$k$ has at most $k - 1$ relevant jobs, so
  \begin{equation*}
    \Delta_{\le x}(s)
    = \Delta_{\le x}(s^-)
    \le \rwork^\B_{\le x}(s^-)
    \le (k - 1)x.
  \end{equation*}

  The other way a many-jobs interval can start is for
  \emph{irrelevant jobs already in System~$k$ to become relevant}.
  For this to happen,
  System~$k$ must be serving $i \ge 1$ irrelevant jobs at~$s^-$.
  Because relevant jobs have priority over irrelevant jobs,
  all relevant jobs must also be in service at~$s^-$.
  There are $i$ irrelevant jobs in service at~$s^-$,
  so there are at most $k - i$ relevant jobs at~$s^-$.
  At time~$s$, at most $i$ irrelevant jobs become relevant,
  so there are at most $k$ relevant jobs at~$s$.
  Each relevant job has size at most~$x$,
  so
  \begin{equation*}
    \Delta_{\le x}(s) \le \rwork^\B_{\le x}(s) \le kx.
    \qedhere
  \end{equation*}
\end{proof}

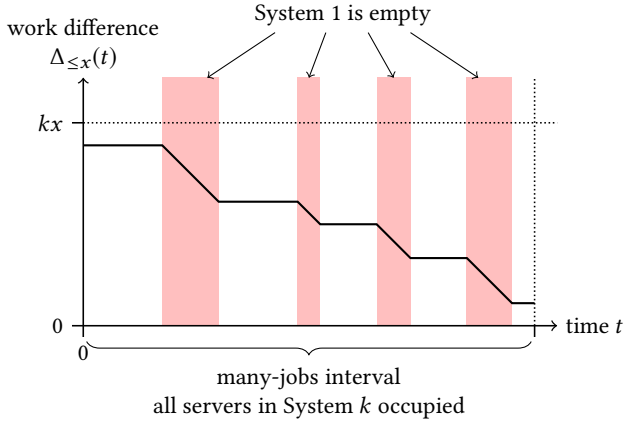
\begin{figure}
  \centering
  \input{fig_delta}

  \caption{Relevant work difference is nonincreasing during many-jobs intervals}
  \label{fig:delta}
\end{figure}

Lemma~\ref{lem:delta} shows that $\Delta_{\le x}(t)$ is bounded at all times.
We can summarize the proof of Lemma~\ref{lem:delta} as follows.
In a few-jobs interval, $\Delta_{\le x}(t)$ is bounded
because there are few relevant jobs in System~$1$
and each contributes a bounded amount of relevant work.
In a many-jobs interval, $\Delta_{\le x}(t)$ is nonincreasing, and hence bounded.

One might intuitively expect $\Delta_{\le x}(t)$ to be constant during a many-jobs interval.
However, $\Delta_{\le x}(t)$ can decrease during a many-jobs interval,
namely when System~$1$ is empty,
as shown in Figure~\ref{fig:delta}.

\subsection{Response Time Bound}

\begin{theorem}
  \label{thm:response_time}
  In an \mgk/,
  the response time of a job of size~$x$ under SRPT\k/ is bounded by
  \begin{equation*}
    T^\srptk(x) \le_\st W^\srpt(x) + B_{\le x}(2kx),
  \end{equation*}
  where $W^\srpt(x)$ denotes the waiting time of a job of size~$x$ under SRPT\1/.
\end{theorem}

\begin{proof}
  From \eqref{eq:busy_bound}, we know that
  \begin{equation*}
  T^\srptk(x)
  \le_\st
  B_{\le x}\Bigl(x + \rwork^\srptk_{\le x} + \vwork^\srptk(x)\Bigr).
  \end{equation*}
  By plugging in
  Lemmas~\ref{lem:vwork} and~\ref{lem:delta},
  we find that
  \begin{align*}
    T^\srptk(x)
    &\le_\st B_{\le x}\bigl(\rwork^\srpt_{\le x} + 2kx\bigr) \\
    &= B_{\le x}\bigl(\rwork^\srpt_{\le x}\bigr) + B_{\le x}(2kx).
  \end{align*}
  Recall from Section~\ref{sub:srpt-1} that the waiting time in SRPT\1/ is
  \begin{equation*}
    W^\srpt(x) = B_{\le x}\bigl(\rwork^\srpt_{\le x}\bigr),
  \end{equation*}
  giving the desired bound.
\end{proof}

While Theorem~\ref{thm:response_time} gives a good bound on the response time under SRPT\k/, we can tighten the bound further by making use of three ideas.
\begin{itemize}
\item
  As the tagged job~$j$ is served,
  its remaining size decreases.
  This decreases the size cutoff for new arrivals to be relevant,
  so not as many arriving jobs contribute to new work.
  Our current bounds to not account for this effect.
\item
  In Lemma~\ref{lem:delta},
  we bound the difference $\Delta_{\le x}(t)$
  between relevant work in System~$1$, which uses SRPT\1/,
  and relevant work in System~$k$, which uses SRPT\k/.
  It turns out that the same proof holds when System~$1$ uses PSJF\1/,
  the preemptive shortest job first policy, instead of SRPT\1/.
  This improves the bound because waiting time under PSJF\1/
  is smaller than waiting time under SRPT\1/ \cite{wierman2005nearly}.
\item
  Even after replacing SRPT\1/ with PSJF\1/,
  Lemma~\ref{lem:delta} is not tight.
  In particular, $\Delta_{\le x}(t)$ is
  at most $x$ times the number of servers serving relevant jobs at time~$t$,
  and there are not always $k$ such servers.
\end{itemize}
These ideas allow us to prove the following tighter bound on mean response time.

\begin{theorem}
  \label{thm:improved-response}
  In an \mgk/,
  the mean response time of a job of size~$x$ under SRPT\k/ is bounded by
  \[\E[\big]{T^\srptk(x)} \le \frac{\int_0^x \lambda t^2f_S(t) \,dt}{2(1-\rho_{\le x})^2}
  + \frac{k\rho_{\le x}x}{1-\rho_{\le x}}
  + \int_0^x \frac{k}{1-\rho_{\le t}} \,dt,\]
  where $f_S(\cdot)$ is the probability density function of
  the service requirement distribution~$S$.
\end{theorem}

\begin{proof}
  See Appendix~\ref{app:improved-srpt}.
  \noqed
\end{proof}

Note that the first term of Theorem~\ref{thm:improved-response}'s upper bound
is the mean waiting time of a job of size~$x$ under PSJF\1/.

\section{Optimality of SRPT\k/ in Heavy Traffic}
\label{sec:optimal}

With the bound derived in Theorem~\ref{thm:response_time}, we can prove our main result on the optimality of SRPT\k/ in the heavy-traffic limit. Theorem~\ref{thm:ratio} will refer to $\E[\big]{T^\srptk}$, which is derived from Theorem~\ref{thm:response_time} by taking the expectation over possible sizes~$x$.

\begin{theorem}
  \label{thm:ratio}
  In an \mgk/ with any service requirement distribution~$S$ which is either
  \begin{enumerate*}[(i)]
  \item
    bounded or
  \item
    unbounded with tail function of upper Matuszewska index\footnote{%
      See Section~\ref{sec:matuszewska} or Appendix~\ref{app:matuszewska}.} less than~$-2$,
  \end{enumerate*}
  \[\lim_{\rho \rightarrow 1} \frac{\E[\big]{T^\srptk}}{\E[\big]{T^\srpt}} = 1.\]
\end{theorem}

To prove Theorem~\ref{thm:ratio}, we start with a result from the literature on the performance of SRPT\1/ in the heavy-traffic limit \cite{Lin}.

\begin{lemma}
  \label{lem:srpt-1-bound}
  In an \mg1/ with any service requirement distribution~$S$ which is either
  \begin{enumerate*}[(i)]
  \item
    bounded or
  \item
    unbounded with tail function of upper Matuszewska index less than~$-2$,
  \end{enumerate*}
  \[\lim_{\rho \rightarrow 1} \frac{\log\Bigl(\frac{1}{1-\rho}\Bigr)}{\E[\big]{T^\srpt}} = 0.\]
\end{lemma}

\begin{proof}
  Follows immediately from results of \citet{Lin}.
  See Appendix~\ref{app:srpt-1-bound}.
  \noqed
\end{proof}

The next step in proving Theorem~\ref{thm:ratio},
is to use the bound on $T^\srptk(x)$ provided by Theorem~\ref{thm:response_time}.
Let $H(x)$ be the bound on $\E[\big]{T^\srptk(x)}$,
\begin{equation}
  \label{eq:H-bound}
  H(x) = \E[\big]{W^\srpt(x) + B_{\le x}(2kx)}.
\end{equation}
By taking the expectation of drawing size~$x$
from the service requirement distribution~$S$,
Theorem~\ref{thm:response_time} implies $\E[\big]{T^\srptk} \le \E{H(S)}$.
The following lemma shows that $\E{H(S)}$ approaches $\E*{T^\srpt}$
in the heavy-traffic limit.

\begin{lemma}
  \label{lem:ratio_for_bound}
  In an \mgk/ with any service requirement distribution~$S$ which is either
  \begin{enumerate*}[(i)]
  \item
    bounded or
  \item
    unbounded with tail function of upper Matuszewska index less than~$-2$,
  \end{enumerate*}
  \[\lim_{\rho \rightarrow 1} \frac{\E{H(S)}}{\E[\big]{T^\srpt}} = 1.\]
\end{lemma}

\begin{proof}
  We know $\E{H(S)} \ge \E[\big]{T^\srptk}$ by Theorem~\ref{thm:response_time},
  and we know $\E[\big]{T^\srptk} \ge \E[\big]{T^\srpt}$ by optimality of SRPT\1/,
  so
  \[\frac{\E{H(S)}}{\E[\big]{T^\srpt}} \ge 1.\]
  We thus only need to show
  \[\lim_{\rho \rightarrow 1} \frac{\E[\big]{H(S)}}{\E[\big]{T^\srpt}} \le 1.\]
  Because $W^\srpt \le T^\srpt$, by \eqref{eq:H-bound} it suffices to show
  \begin{equation}
    \label{eq:busy_goal}
    \lim_{\rho \rightarrow 1} \frac{\E{B_{\le S}(2kS)}}{\E[\big]{T^\srpt}} = 0.
  \end{equation}
  Applying standard results for busy periods \cite{harchol2013performance},
  \begin{equation*}
    \E*{B_{\le S}(2kS)}
    = 2k\E*{\frac{S}{1-\rho_{\le S}}}
    = 2k\int_0^\infty \frac{x f_S(x)}{1-\rho_{\le x}} \,dx,
  \end{equation*}
  where $f_S(\cdot)$ is the probability density function of~$S$.
  To compute the integral,
  we make a change of variables from $x$ to $\rho_{\le x}$
  (see Definition~\ref{def:busy_period}),
  which uses the following facts:
  \begin{align*}
    \rho_{\le x} &= \lambda \E{S\indicator(S < x)} = \int_0^x \lambda t f_S(t) \,dt \\
    \frac{d\rho_{\le x}}{dx} &= \lambda x f_S(x) \\
    \rho_{\le 0} &= 0 \\
    \lim_{x \rightarrow \infty} \rho_{\le x} &= \rho.
  \end{align*}
  Given this change of variables, we compute
  \begin{align*}
    \E*{\frac{S}{1-\rho_{\le S}}} &= \int_0^\infty \frac{x f_S(x)}{1-\rho_{\le x}} \,dx \\
                                  &= \int_0^\rho \frac{1}{\lambda(1 - \rho_{\le x})}\,d\rho_{\le x} \\
                                  &= \frac{1}{\lambda}\ln \left(\frac{1}{1-\rho}\right) \\
                                  &= \Theta\left (\log \left (\frac{1}{1-\rho} \right) \right).
  \end{align*}
  This means $\E*{B_{\le S}(2kS)} = \Theta(\log(1/(1-\rho)))$,
  so \eqref{eq:busy_goal} follows from Lemma~\ref{lem:srpt-1-bound}.
\end{proof}

Armed with Theorem~\ref{thm:response_time} and Lemma~\ref{lem:ratio_for_bound}, we are now prepared to prove our main result, Theorem~\ref{thm:ratio}.
\begin{proof}[Proof of Theorem~\ref{thm:ratio}]
  Because SRPT\1/ minimizes mean response time, it suffices to show that
  \begin{equation*}
    \lim_{\rho \rightarrow 1} \frac{\E[\big]{T^\srptk}}{\E*{T^\srpt}} \le 1,
  \end{equation*}
  which follows immediately from Theorem~\ref{thm:response_time}
  and Lemma~\ref{lem:ratio_for_bound}.
\end{proof}

Theorem~\ref{thm:ratio} and the optimality of SRPT\1/ imply that
SRPT\k/ is optimal in the heavy-traffic limit.

\begin{corollary}
  \label{cor:optimal}
  In an \mgk/ with any service requirement distribution~$S$ which is either
  \begin{enumerate*}[(i)]
  \item
    bounded or
  \item
    unbounded with tail function of upper Matuszewska index less than~$-2$,
  \end{enumerate*}
  \begin{equation*}
    \lim_{\rho \rightarrow 1} \frac{\E[\big]{T^\srptk}}{\E*{T^P}} \le 1
  \end{equation*}
  for any scheduling policy~$P$.
\end{corollary}

Recall from \eqref{eq:H-bound}
that Theorem~\ref{thm:response_time} implies $\E[\big]{T^\srptk} \le \E{H(S)}$.
Similarly, letting
\begin{equation*}
  I(x) =
  \frac{\int_0^x \lambda t^2f_S(t) \,dt}{2(1-\rho_{\le x})^2}
  + \frac{k\rho_{\le x}x}{1-\rho_{\le x}}
  + \int_0^x \frac{k}{1-\rho_{\le t}} \,dt,
\end{equation*}
Theorem~\ref{thm:improved-response} implies $\E[\big]{T^\srptk} \le \E{I(S)}$.
Lemma~\ref{lem:ratio_for_bound} and the optimality of SRPT\1/ imply that
these bounds on SRPT\k/'s mean response time are tight as $\rho \to 1$.

\begin{corollary}
  \label{cor:tight}
  In an \mgk/ with any service requirement distribution~$S$ which is either
  \begin{enumerate*}[(i)]
  \item
    bounded or
  \item
    unbounded with tail function of upper Matuszewska index less than~$-2$,
  \end{enumerate*}
  \begin{equation*}
    \lim_{\rho \rightarrow 1} \frac{\E{H(S)}}{\E*{T^\srptk}}
    = \lim_{\rho \rightarrow 1} \frac{\E{I(S)}}{\E*{T^\srptk}}
    = 1.
  \end{equation*}
\end{corollary}

\begin{proof}
  After applying Theorem~\ref{thm:response_time},
  Lemma~\ref{lem:ratio_for_bound},
  and the optimality of SRPT\1/,
  we know that
  \begin{equation*}
    \lim_{\rho \rightarrow 1} \frac{\E{H(S)}}{\E*{T^\srptk}}
    = 1.
  \end{equation*}
  All that remains is to show $I(x) \le H(x)$.
  This holds because
  \begin{equation*}
    \frac{\int_0^x \lambda t^2f_S(t) \,dt}{2(1-\rho_{\le x})^2}
    \le
    \E[\big]{W^\srpt(x)}
  \end{equation*}
  by the standard analysis of $W^\srpt(x)$ \cite{schrage1966queue},
  and
  \begin{equation*}
    \frac{k\rho_{\le x}x}{1-\rho_{\le x}} + \int_0^x \frac{k}{1-\rho_{\le t}} \,dt
    \le
    \frac{2kx}{1-\rho_{\le x}}
    =
    \E{B_{\le x}(2kx)}.
    \qedhere
  \end{equation*}
\end{proof}

\begin{figure*}
  \begin{subfigure}[t]{0.475\textwidth}
  \centering
  Ratio $\E[\big]{T^\srptk}/\E[\big]{T^\srpt}$

  \includegraphics[width=\linewidth,trim={0 0.3cm 0 0.5cm},clip]{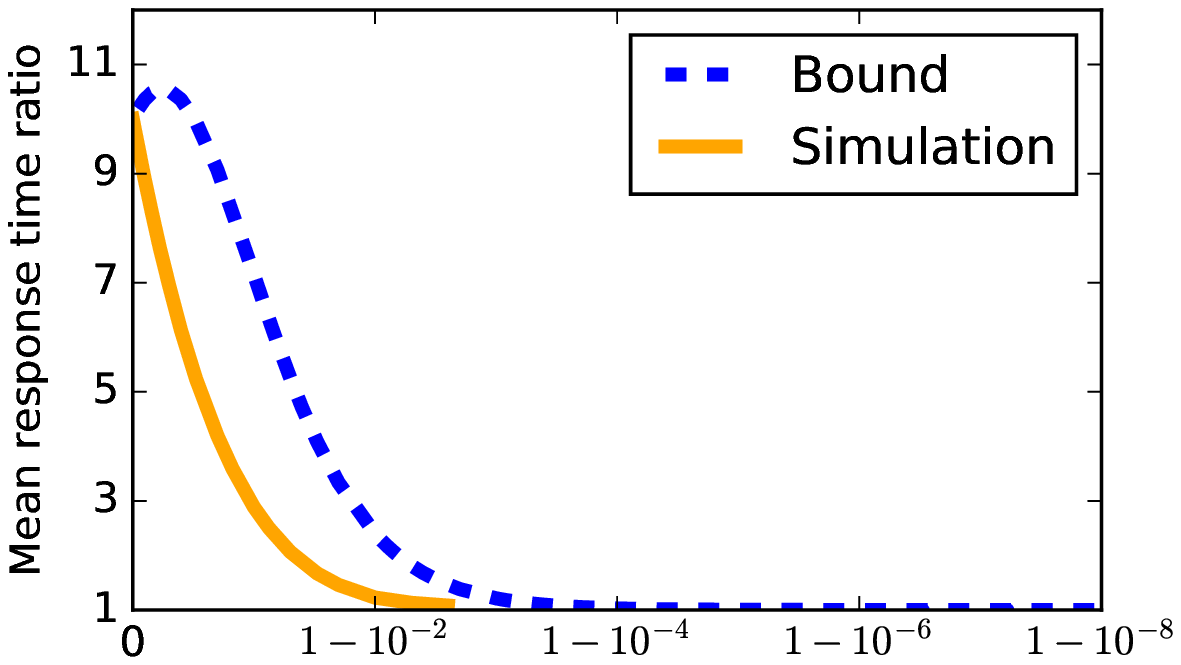}

  System load ($\rho$)
  \end{subfigure}
  \begin{subfigure}[t]{0.475\textwidth}
    \centering
  Ratio $\E[\big]{T^\srptk}/\E[\big]{T^\srpt}$

  \includegraphics[width=\linewidth,trim={0 0.3cm 0 0.5cm},clip]{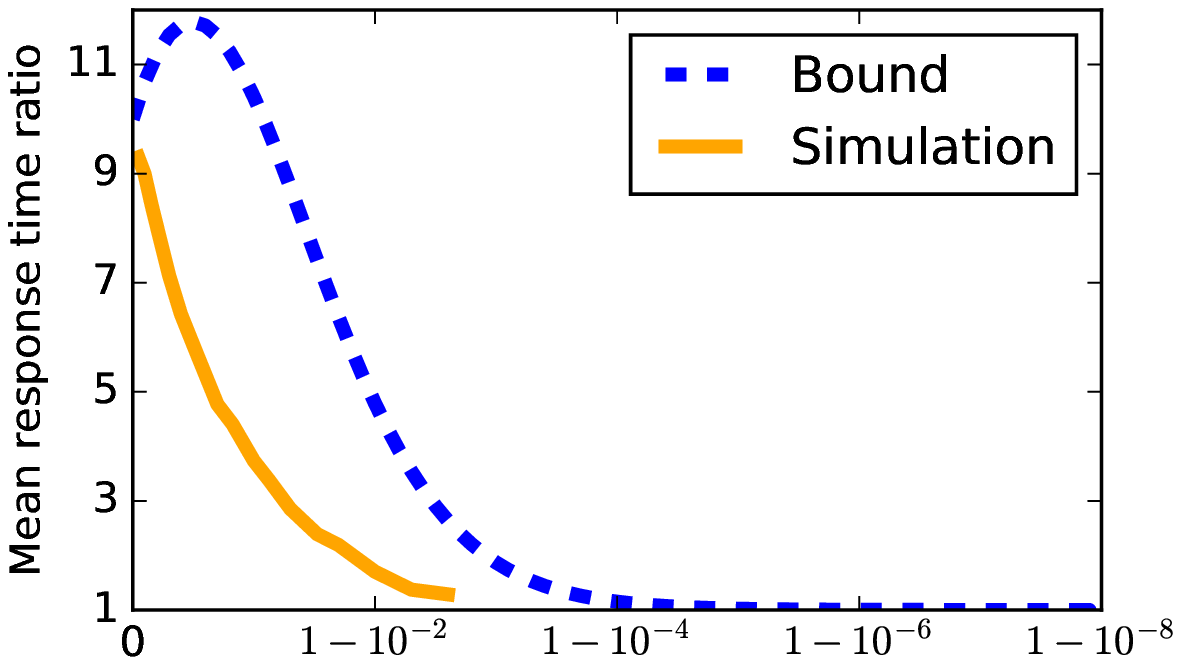}

  System load ($\rho$)
  \end{subfigure}
  \vspace{0.5\baselineskip}

  \captiondetail
  The plots above show the ratio $\E[\big]{T^\srptk}/\E[\big]{T^\srpt}$.
  Observe that as $\rho \to 1$, both our bound and the simulation
  converge to a ratio of 1.
  Our simulations of this ratio are the solid orange curves.
  Our analytic upper bounds derived in Theorem~\ref{thm:improved-response} are the dashed blue curves.
  We use $k = 10$ servers. The service requirement distribution~$S$ is $\mathrm{Uniform}(0, 2)$
  in the left plot and a Hyperexponential distribution with $E[S]=1$ and $C^2 = 10$ in the right plot.
  We only simulate up to $\rho=0.9975$ due to long convergence times.

  \caption{Convergence of mean response time ratio}
  \label{fig:ratio}
\end{figure*}

As an illustration of the optimality of SRPT\k/, we plot the ratio $\E[\big]{T^\srptk}/\E*{T^\srpt}$ in Figure~\ref{fig:ratio}. The solid orange lines show simulation results for this ratio. For the dashed blue lines, we used our analysis from Theorem~\ref{thm:improved-response} as an upper bound on $\E[\big]{T^\srptk}$, and divided by the known results for $\E[\big]{T^\srpt}$. The important feature to notice in Figure~\ref{fig:ratio} is that as system load $\rho$ approaches~$1$, both our analytic bound and the simulation converge to~$1$.

\section{Other Scheduling Policies}
\label{sec:other-policies}
We generalize our analysis to give the first response time bounds on several additional multiserver scheduling policies. Using the bounds, we prove optimality results for each policy as $\rho \rightarrow 1$. For each policy~$P$, we generalize the usual single-server policy, written $P$\1/, to a multiserver policy for $k$ servers, written $P$\k/, by preemptively serving the $k$ jobs with highest priority at any time.
\begin{itemize}
\item \emph{Preemptive Shortest Job First} (PSJF) prioritizes the jobs with smallest \emph{original} size.
  PSJF achieves performance comparable to SRPT despite not tracking every job's age \cite{harchol2013performance}.
\item \emph{Remaining Size Times Original Size} (RS) prioritizes the jobs with the smallest product of original size and remaining size. RS is also known as Size Processing Time Product (SPTP).
  RS is optimal for minimizing \emph{mean slowdown} \cite{Hyytia:2012:MSH:2254756.2254763}.
\item \emph{Foreground-Background} (FB) prioritizes the jobs with smallest \emph{age},
  meaning the jobs that have been served the least so far.
  FB is also known as Least Attained Service (LAS).
  When the service requirement distribution has decreasing hazard rate,
  FB minimizes mean response time among all scheduling policies
  that do not have access to job sizes \cite{righter1989scheduling}.
\end{itemize}
We give the first response time bounds for PSJF\k/, RS\k/ and FB\k/.
We then use these bounds to prove the following optimality results,
under mild assumptions on the service requirement distribution:
\begin{itemize}
\item
  In the $\rho \to 1$ limit,
  PSJF\k/ and RS\k/ minimize mean response time among all scheduling policies
  (see Theorems~\ref{thm:psjf-optimal} and~\ref{thm:rs-optimal}).
\item
  In the $\rho \to 1$ limit,
  FB\k/ minimizes mean response time under the same conditions as FB\1/
  (see Theorem~\ref{thm:fb-optimal}).
\end{itemize}

Our analyses follow the same steps as in Section~\ref{sec:analysis}.
\begin{itemize}
\item Use the four categories of work to bound
  the response time of the tagged job~$j$
  in terms of virtual work and steady-state relevant work.
\item Bound virtual work.
\item Bound steady-state relevant work.
\end{itemize}
Because different scheduling policies prioritize jobs differently,
we use a different definition of ``relevant jobs'' for each policy.
Under PSJF\k/ and RS\k/,
the definition of relevant jobs is very similar to that for SRPT\k/,
allowing us to use familiar tools such as relevant busy periods $B_{\le x}(\cdot)$.
However, FB\k/ uses a somewhat different definition of relevant jobs,
resulting in a few changes to the analysis.

Finally, in Section~\ref{sub:fcfs},
we discuss why our technique \emph{does not} generalize to
the First-Come, First-Served (FCFS) scheduling policy.

\subsection{Preemptive Shortest Job First (PSJF\k/)}
\label{sec:psjf}
As usual, we consider a tagged job~$j$ of size~$x$.
Under PSJF\k/, another job~$\ell$ is \emph{relevant} to~$j$
if $\ell$ has \emph{original} size at most~$x$.
With this definition of relevance,
we divide work into the same four categories as in Section~\ref{sec:analysis},
namely tagged, old, new, and virtual.
This bounds the response time of $j$ by
\begin{equation}
  \label{eq:psjf-busy}
  T^\psjfk \le_\st
  B_{\le x}\Bigl(x + \rwork^\psjfk_{\le x} + \vwork^\psjfk(x)\Bigr).
\end{equation}

The proof of Lemma~\ref{lem:vwork} works nearly verbatim for PSJF\k/, so
\begin{equation}
  \label{eq:psjf-vwork}
  \vwork^\psjfk(x) \le (k - 1)x.
\end{equation}

The analysis of steady-state relevant work
is similar to that in Section~\ref{sec:srpt-rwork}.
We consider a pair of systems experiencing the same arrival sequence:
System~$1$, which uses PSJF\1/,
and System~$k$, which uses PSJF\k/.
We define $\Delta^\psjfk_{\le x}(t)$ to be
the difference between the amounts of relevant work in the two systems
at time~$t$.
We then bound $\Delta^\psjfk_{\le x}(t)$.

\begin{lemma}
  \label{lem:psjf-delta}
  The difference in relevant work between Systems~$1$ and~$k$ is bounded by
  \begin{equation*}
    \Delta^\psjfk_{\le x}(t) \le (k - 1)x.
  \end{equation*}
\end{lemma}

\begin{proof}
  We define few-jobs intervals and many-jobs intervals as in Section~\ref{sec:srpt-rwork}.
  The case where $t$ is in a few-jobs interval is simple:
  there are at most $k - 1$ relevant jobs in System~$k$ at time~$t$,
  each of remaining size at most~$x$, so
  \begin{equation*}
    \Delta^\psjfk_{\le x}(t) \le (k - 1)x.
  \end{equation*}

  Suppose instead that $t$ is in a many-jobs interval.
  Let time~$s$ be the start of the many-jobs interval containing~$t$.
  By essentially the same argument as in the proof of Lemma~\ref{lem:delta},\footnote{%
    In fact, the argument for PSJF is slightly simpler than that for SRPT,
    because irrelevant jobs never become relevant under PSJF.}
  \begin{equation*}
    \Delta^\psjfk_{\le x}(t) \le \Delta^\psjfk_{\le x}(s).
  \end{equation*}
  It thus suffices to show $\Delta^\psjfk_{\le x}(s) \le (k - 1)x$.
  The only way a many-jobs interval can start under PSJF\k/
  is for a relevant job to arrive while System~$k$ has $k - 1$ relevant jobs.
  The same arrival occurs in System~$1$, so
  \begin{equation*}
    \Delta^\psjfk_{\le x}(s) = \Delta^\psjfk_{\le x}(s^-) \le (k - 1)x
  \end{equation*}
  because $s^-$, the instant before~$s$, is in a few-jobs interval.
\end{proof}

\begin{theorem}
  \label{thm:psjf}
  In an \mgk/,
  the response time of a job of size $x$ under PSJF\k/ is bounded by
  \begin{equation*}
    T^\psjfk(x) \le_\st W^\psjf(x) + B_{\le x}((2k - 1)x).
  \end{equation*}
\end{theorem}

\begin{proof}
  By \eqref{eq:psjf-busy}, \eqref{eq:psjf-vwork}, and Lemma~\ref{lem:psjf-delta},
  \begin{align*}
    T^\psjfk(x)
    &\le_\st B_{\le x}\bigl(\rwork^\psjf_{\le x} + (2k - 1)x\bigr) \\
    &= B_{\le x}\bigl(\rwork^\psjf_{\le x}\bigr) + B_{\le x}(2k - 1).
  \end{align*}
  The waiting time in PSJF\1/ is
  \begin{equation*}
    W^\psjf(x) = B_{\le x}\bigl(\rwork_{\le x}^\psjf\bigr),
  \end{equation*}
  giving the desired bound.
\end{proof}

With the bound derived in Theorem~\ref{thm:psjf}, we can prove that PSJF\k/ also minimizes mean response time in the heavy-traffic limit.

\begin{theorem}
  \label{thm:psjf-optimal}
    In an \mgk/ with any service requirement distribution~$S$ which is either
  \begin{enumerate*}[(i)]
  \item
    bounded or
  \item
    unbounded with tail function of upper Matuszewska index\footnote{%
      See Section~\ref{sec:matuszewska} or Appendix~\ref{app:matuszewska}.} less than~$-2$,
  \end{enumerate*}
  \[\lim_{\rho \rightarrow 1} \frac{\E[\big]{T^\psjfk}}{\E*{T^\srpt}} = 1.\]
\end{theorem}
\begin{proof}
  From Theorem~\ref{thm:psjf}, we know that
  \begin{equation*}
    T^\psjfk(x) \le_\st W^\psjf(x) + B_{\le x}((2k - 1)x)
  \end{equation*}
  However, $W^\psjf(x) \le_\st W^\srpt(x)$ \cite{wierman2005nearly}. Therefore,
  \begin{align*}
    T^\psjfk(x) &\le_\st W^\srpt(x) + B_{\le x}((2k - 1)x)\\
               &\le_\st W^\srpt(x) + B_{\le x}(2kx)
  \end{align*}
  This bound on $T^\psjfk(x)$ is the same as the bound on $T^\srptk(x)$ given in Theorem~\ref{thm:response_time}. The rest of the proof proceeds as in the proof of Theorem~\ref{thm:ratio}.
\end{proof}

As in Corollary~\ref{cor:optimal}, Theorem~\ref{thm:psjf-optimal} and the optimality of SRPT\1/ imply that PSJF\k/ is optimal in the heavy-traffic limit.

\subsection{Remaining Size Times Original Size (RS\k/)}

As usual, we consider a tagged job~$j$ of size~$x$.
When $j$ has remaining size~$y$,
another job~$\ell$ is \emph{relevant} to~$j$
if the \emph{product of $\ell$'s original size and remaining size}
is at most~$xy$.
In particular, if $\ell$ is relevant to~$j$,
then $\ell$'s remaining size is at most~$x$.
With this definition of relevance,
we divide work into the same four categories as in Section~\ref{sec:analysis},
namely tagged, old, new, and virtual.
This bounds the response time of $j$ by
\begin{equation}
  \label{eq:rs-busy}
  T^\rsk \le_\st
  B_{\le x}\Bigl(x + \rwork^\rsk_{\le x} + \vwork^\rsk(x)\Bigr).
\end{equation}

The proof of Lemma~\ref{lem:vwork} works nearly verbatim for RS\k/, so
\begin{equation}
  \label{eq:rs-vwork}
  \vwork^\rsk(x) \le (k - 1)x.
\end{equation}

The analysis of steady-state relevant work
is similar to that in Section~\ref{sec:srpt-rwork}.
We consider a pair of systems experiencing the same arrival sequence:
System~$1$, which uses RS\1/,
and System~$k$, which uses RS\k/.
We define $\Delta^\rsk_{\le x}(t)$ to be
the difference between the amounts of relevant work in the two systems
at time~$t$.
We then bound $\Delta^\rsk_{\le x}(t)$.

\begin{lemma}
  \label{lem:rs-delta}
  The difference in relevant work between Systems~$1$ and~$k$ is bounded by
  \begin{equation*}
    \Delta^\rsk_{\le x}(t) \le kx.
  \end{equation*}
\end{lemma}

\begin{proof}
  Even though RS uses a definition of relevant jobs different from SRPT's,
  the proof is analogous to that of Lemma~\ref{lem:delta}.
\end{proof}

\begin{theorem}
  \label{thm:rs}
  In an \mgk/,
  the response time of a job of size $x$ under RS\k/ is bounded by
  \begin{equation*}
    T^\rsk(x) \le_\st W^\rs(x) + B_{\le x}((2k - 1)x).
  \end{equation*}
\end{theorem}

\begin{proof}
  By \eqref{eq:rs-busy}, \eqref{eq:rs-vwork}, and Lemma~\ref{lem:rs-delta},
  \begin{align*}
    T^\rsk(x)
    &\le_\st B_{\le x}\bigl(\rwork^\rs_{\le x} + 2kx\bigr) \\
    &\le_\st B_{\le x}\bigl(\rwork^\rs_{\le x}\bigr) + B_{\le x}(2kx).
  \end{align*}
  The waiting time in RS\1/ is
  \begin{equation*}
    W^\rs(x) = B_{\le x}\bigl(\rwork_{\le x}^\psjf\bigr),
  \end{equation*}
  giving the desired bound.
\end{proof}

With the bound derived in Theorem~\ref{thm:rs}, we can prove that RS\k/ also minimizes mean response time in the heavy-traffic limit.

\begin{theorem}
  \label{thm:rs-optimal}
  In an \mgk/ with any service requirement distribution~$S$ which is either
  \begin{enumerate*}[(i)]
  \item
    bounded or
  \item
    unbounded with tail function of upper Matuszewska index\footnote{%
      See Section~\ref{sec:matuszewska} or Appendix~\ref{app:matuszewska}.} less than~$-2$,
  \end{enumerate*}
  \[\lim_{\rho \rightarrow 1} \frac{\E[\big]{T^\rsk}}{\E*{T^\srpt}} = 1.\]
\end{theorem}
\begin{proof}
  From Theorem~\ref{thm:rs}, we know that
  \begin{equation*}
    T^\rsk(x) \le_\st W^\rs(x) + B_{\le x}(2kx)
  \end{equation*}
  However, $W^\rs(x) \le_\st W^\srpt(x)$ \cite{wierman2005nearly}. Therefore,
  \begin{equation*}
    T^\rsk(x) \le_\st W^\srpt(x) + B_{\le x}(2kx)
  \end{equation*}
  This bound on $T^\rsk(x)$ is the same as the bound on $T^\srptk(x)$ given in Theorem~\ref{thm:response_time}. The rest of the proof proceeds as in the proof of Theorem~\ref{thm:ratio}.
\end{proof}

As in Corollary~\ref{cor:optimal}, Theorem~\ref{thm:rs-optimal} and the optimality of SRPT\1/ imply that RS\k/ is optimal in the heavy-traffic limit.

We have so far shown response time bounds for SRPT\k/, PSJF\k/, and RS\k/
that are strong enough to prove asymptotic optimality in heavy traffic.
We conjecture that similar bounds and optimality results
hold for multiserver variants of
any policy in the SMART class \cite{wierman2005nearly},
which includes SRPT, PSJF, and RS.

\subsection{Foreground-Background (FB\k/)}

The analysis of FB\k/ proceeds similarly to the analysis of SRPT\k/
but with a few more changes than were needed for PSJF\k/ and RS\k/.
To analyze PSJF\k/ and RS\k/,
we followed the same outline as Section~\ref{sec:analysis}
with a small change to the definition of relevant jobs.
In particular,
we reused the notion of relevant busy periods $B_{\le x}(\cdot)$
from Definition~\ref{def:busy_period}.
In contrast, as we will see shortly,
FB\k/ has a significantly different definition of relevant jobs,
so the definition of relevant busy periods will also change.

As usual, we consider a tagged job~$j$ of size~$x$.
Recall that FB prioritizes the jobs of smallest \emph{age}, or attained service.
When $j$ arrives, its age is~$0$,
so it has priority over all other jobs in the system.
However, as $j$ is served, its age increases and its priority gets worse.
The key to the usual single-server analysis of FB is that
to define relevant work, we have to look at $j$'s \emph{worst future priority}
\cite{schrage1967queue, harchol2013performance, scully2018soap}.
This worst priority occurs when $j$ has age~$x$, an instant before completion,
giving us the following definition of relevant jobs.

\begin{definition}
  \label{def:relevant_fb}
  Suppose job~$j$ has original size~$x$.
  Under FB\k/, a job~$\ell$ is \emph{relevant} to job~$j$
  if $\ell$ has age at most~$x$.
  Otherwise $\ell$ is \emph{irrelevant} to~$j$.
\end{definition}

There is an important difference between the notions of relevance
for SRPT\k/ and FB\k/.
Under SRPT\k/,
each arriving job starts as either relevant or irrelevant to~$j$
and remains that way for $j$'s entire time in the system.
In contrast, under FB\k/,
\emph{every new arrival is at least temporarily relevant to~$j$}.
Specifically, if a new arrival~$\ell$ has size at most~$x$,
then $\ell$ is relevant to~$j$ for its entire time in the system.
If $\ell$ instead has size greater than~$x$,
then $\ell$ is relevant to~$j$ only until it reaches age~$x$,
at which point it becomes irrelevant.
This observation motivates the definition of relevant busy periods for FB\k/.

\begin{definition}
  \label{def:busy_period_fb}
  Under FB\k/,
  a \emph{relevant busy period} for a job of size~$x$ started by
  (possibly random) amount of work~$V$,
  written~$B_{\bar x}(V)$,
  is the amount of time it takes for
  a work-conserving system that starts with $V$ work to become empty,
  where \emph{every arrival's service is truncated at age~$x$}.
  A relevant busy period has expectation
  \begin{equation*}
    \E{B_{\bar x}(V)} = \frac{\E{V}}{1 - \rho_{\bar x}}.
  \end{equation*}
  Above, $\rho_{\bar x}$ is the \emph{relevant load} for a job of size~$x$,
  which is the total load due to relevant jobs.
  Its value is
  \begin{equation*}
    \rho_{\bar x} = \lambda\E{\min(S, x)},
  \end{equation*}
  because each arrival is relevant only until it reaches age~$x$.
\end{definition}

We make a similar modification to the definition of steady-state relevant work.

\begin{definition}
  \label{def:rwork_fb}
  The \emph{steady-state relevant work} for a job of size~$x$ under FB\k/,
  written $\rwork^\fbk_{\bar x}$,
  is the sum of \emph{remaining truncated sizes} of all jobs
  observed at a random point in time.
  A job's remaining truncated size is the amount of time until it either
  completes or reaches age~$x$.
\end{definition}

Armed with Definitions~\ref{def:relevant_fb},
\ref{def:busy_period_fb}, and~\ref{def:rwork_fb},
we divide work into the same four categories as in Section~\ref{sec:analysis},
namely tagged, old, new, and virtual.
This bounds the response time of $j$ by
\begin{equation}
  \label{eq:fb-busy}
  T^\fbk \le_\st
  B_{\bar x}\Bigl(x + \rwork^\fbk_{\bar x} + \vwork^\fbk(x)\Bigr).
\end{equation}

The proof of Lemma~\ref{lem:vwork} works nearly verbatim for FB\k/, so
\begin{equation}
  \label{eq:fb-vwork}
  \vwork^\fbk(x) \le (k - 1)x.
\end{equation}

The analysis of steady-state relevant work
is similar to that in Section~\ref{sec:srpt-rwork}.
We consider a pair of systems experiencing the same arrival sequence:
System~$1$, which uses FB\1/,
and System~$k$, which uses PSJF\k/.
We define $\Delta^\fbk_{\bar x}(t)$ to be
the difference between the amounts of relevant work in the two systems
at time~$t$.
We then bound $\Delta^\fbk_{\bar x}(t)$.

\begin{lemma}
  \label{lem:fb-delta}
  The difference in relevant work between Systems~$1$ and~$k$ is bounded by
  \begin{equation*}
    \Delta^\fbk_{\bar x}(t) \le (k - 1)x.
    \end{equation*}
\end{lemma}

\begin{proof}
  Even though FB uses a definition of relevant jobs
  different from PSJF's,\footnote{%
    We draw an analogy with PSJF rather than SRPT because
    under both FB and PSJF, irrelevant jobs never become relevant.}
  the proof is analogous to that of Lemma~\ref{lem:psjf-delta}.
\end{proof}

\begin{theorem}
  \label{thm:fb}
  In an \mgk/,
  the response time of a job of size $x$ under FB\k/ is bounded by
  \begin{equation*}
    T^\fbk(x) \le_\st B_{\bar x}\bigl(\rwork^\fbk_{\bar x} + (2k - 1)x\bigr).
  \end{equation*}
\end{theorem}

\begin{proof}
  Combining \eqref{eq:fb-busy}, \eqref{eq:fb-vwork}, and Lemma~\ref{lem:fb-delta}
  yields the desired bound.
\end{proof}

Note that the waiting time under FB\1/ is always zero, as a new job immediately receives service, so we do not phrase the bound in terms of waiting time.

With the bound derived in Theorem~\ref{thm:rs}, we can prove that the mean response time of FB\k/ approaches that of FB\1/ in the heavy-traffic limit.
We make use of prior work on
the mean response time of FB in heavy traffic \cite{kamphorst2017heavy}.
Let
\begin{align*}
  W(x) &= \E{B_{\bar x}(\rwork^\fbk_{\bar x})} \\
  R(x) &= \E{B_{\bar x}(x)}.
\end{align*}
$W(x)$ and $R(x)$ are not the mean waiting and residence times
of a job of size~$x$ under FB
because waiting time is always zero,
but they play roughly analogous roles
in the standard analysis of FB \cite[Section~5]{scully2018soap}.

\begin{lemma}
  \label{lem:fb-1-bound}
  In an \mg1/ with any service requirement distribution~$S$ which
  is unbounded with tail function of upper Matuszewska index\footnote{%
      See Section~\ref{sec:matuszewska} or Appendix~\ref{app:matuszewska}.} less than~$-2$,
  \begin{equation*}
    \lim_{\rho \rightarrow 1} \frac{\E{R(S)}}{\E[\big]{T^\fb}} = 0.
  \end{equation*}
\end{lemma}

\begin{proof}
  Follows immediately from results of \citet{kamphorst2017heavy}.
  See Appendix~\ref{app:fb-1-bound}.
  \noqed
\end{proof}

\begin{theorem}
  \label{thm:fb-optimal}
  In an \mgk/ with any service requirement distribution~$S$ which
  is unbounded with tail function of upper Matuszewska index less than~$-2$,
  \begin{equation*}
    \lim_{\rho \rightarrow 1} \frac{\E[\big]{T^\fbk}}{\E[\big]{T^\fb}} = 1.
  \end{equation*}
\end{theorem}

\begin{proof}
  The standard analysis of FB\1/
  \cite{schrage1967queue, harchol2013performance}
  shows
  \begin{equation*}
    \E[\big]{T^\fb} = \E{W(S)} + \E{R(S)},
  \end{equation*}
  whereas Theorem~\ref{thm:fb} implies
  \begin{equation*}
    \E[\big]{T^\fbk} \le \E{W(S)} + (2k - 1)\E{R(S)},
  \end{equation*}
  so the result follows by Lemma~\ref{lem:fb-1-bound}.
\end{proof}

\Citet{righter1989scheduling} show that
when the job size distribution~$S$ has decreasing hazard rate,
FB\1/ is optimal for minimizing response time among all scheduling policies
that do not have access to job sizes.
Theorem~\ref{thm:fb-optimal} implies that in the heavy-traffic limit,
FB\k/ is optimal in the same setting.\footnote{%
  It has been claimed that FB\k/ is optimal for arbitrary arrival sequences
  when the service requirement distribution has decreasing hazard rate
  \cite[Theorem~2.1]{wu2004scheduling}.
  However, the proof has an error.
  See Appendix~\ref{app:interchange}.}

\begin{corollary}
  In an \mgk/ with any service requirement distribution~$S$ which
  \begin{enumerate*}[(a)]
  \item
    is unbounded,
  \item
    has decreasing hazard rate, and
  \item
    has tail function of upper Matuszewska index less than~$-2$,
  \end{enumerate*}
  \begin{equation*}
    \lim_{\rho \rightarrow 1} \frac{\E[\big]{T^\fb}}{\E*{T^P}} \le 1
  \end{equation*}
  for any scheduling policy~$P$ that does not have access to job sizes.
\end{corollary}

\subsection{What about First-Come, First-Served?}
\label{sub:fcfs}

Having seen the success of our modified tagged job analysis for a variety of policies,
it is natural to ask: does a similar analysis work for
the multiserver First-Come, First-Served policy (FCFS\k/)?

Unfortunately, our technique does not work for FCFS\k/.
To see why, let us take a look at what our analyses of
SRPT\k/, PSJF\k/, RS\k/, and FB\k/ have in common.
A central component of all four analyses is bounding
the \emph{difference in relevant work}
between two systems experiencing the same arrival sequence,
one using a single-server policy $P$\1/
and another using its $k$-server variant $P$\k/.
These bounds are given in Lemmas~\ref{lem:delta}, \ref{lem:psjf-delta},
\ref{lem:rs-delta}, and~\ref{lem:fb-delta}.
All four lemmas have similar two-step proofs.
\begin{itemize}
\item
  First, they bound the \emph{number of relevant jobs}
  both during few-jobs intervals
  and at the start of many-jobs intervals.
  For all four policies, this bound is at most~$k$.
\item
  Second, they bound the \emph{relevant work contributed by each relevant job}.
  For all four policies, this bound is~$x$.
\end{itemize}
When we try to prove analogous bounds for FCFS\k/,
we can still bound the number of relevant jobs by~$k$,
but \emph{the relevant work contributed by each relevant job is unbounded}.

The definition of relevant jobs
is the crucial difference between FCFS\k/ and the policies we analyze.
Consider the jobs relevant to a tagged job~$j$ of size~$x$.
\begin{itemize}
\item
  Under SRPT\k/, PSJF\k/, and RS\k/,
  only \emph{some} jobs are relevant to~$j$,
  and all such jobs have size at most~$x$.
\item
  Under FB\k/, while all jobs might be relevant to~$j$,
  they are only \emph{temporarily} relevant,
  each contributing at most $x$ relevant work.
\item
  However, under FCFS\k/,
  \emph{all} jobs in the system when $j$ arrives
  are \emph{permanently} relevant to~$j$.
\end{itemize}
This means that if the service requirement distribution~$S$ is unbounded,
our worst-case technique is insufficient for
bounding the difference in relevant work between FCFS\1/ and FCFS\k/.

\section{Conclusion}

We give the first stochastic bound on
the response time of SRPT\k/ (see Section~\ref{sec:analysis}).
Using this bound, we show that
SRPT\k/ has asymptotically optimal mean response time
in the heavy-traffic limit (see Section~\ref{sec:optimal}).
We generalize our analysis
to give the first stochastic bounds on the response times of
the PSJF\k/, RS\k/ and FB\k/ policies,
and we use these bounds to prove asymptotic optimality results
for all three policies (see Section~\ref{sec:other-policies}).

To achieve these results,
we strategically combine stochastic and worst-case techniques.
Specifically, we obtain our bounds using a modified tagged job analysis.
Traditional tagged job analyses for single-server systems
rely on properties that do not hold in multiserver systems,
notably work conservation.
To make tagged job analysis work for multiple servers,
we use two key insights.
\begin{itemize}
\item
  We introduce the concept of \emph{virtual work}
  (see Section~\ref{sec:analysis}),
  which makes the system appear work-conserving
  while the tagged job is in the system.
  We give a worst-case bound for virtual work.
\item
  We compare the multiserver system with a
  \emph{single-server system of the same service capacity}.
  We show that even in the worst case,
  the steady state amount of relevant work under SRPT\k/
  is close to the steady state amount of relevant work under SRPT\1/.
\end{itemize}
Applying these two insights to the tagged job analysis
gives a \emph{stochastic} expression bounding response time.

One direction for future work is to apply our technique
to a broader range of scheduling policies.
In particular, we conjecture that out results generalize
to the SMART class of policies \cite{wierman2005nearly},
which includes SRPT, PSJF, and RS.
Another direction is to improve our response time bounds under low system load.
While our bounds are valid for all loads,
they are only tight for load near capacity.

\bibliographystyle{ACM-Reference-Format}
\bibliography{srpt}

\appendix

\section{Improved SRPT\k/ Bound}
\label{app:improved-srpt}

\begin{reptheorem}{thm:improved-response}
  In an \mgk/,
  the mean response time of a job of size~$x$ under SRPT\k/ is bounded by
  \[\E[\big]{T^\srptk(x)} \le \frac{\int_0^x \lambda t^2f_S(t) \,dt}{2(1-\rho_{\le x})^2}
  + \frac{k\rho_{\le x}x}{1-\rho_{\le x}}
  + \int_0^x \frac{k}{1-\rho_{\le t}} \,dt,\]
  where $f_S(\cdot)$ is the probability density function of
  the service requirement distribution~$S$.
\end{reptheorem}

\begin{proof}
  We will prove Theorem~\ref{thm:improved-response} by proving improved versions
  of \eqref{eq:busy_bound} and Lemma~\ref{lem:delta}.

  A key element of our analysis is bounding the amount of new work done
  while the tagged job $j$ of size $x$ is in the system.
  In \eqref{eq:busy_bound},
  we bound this quantity by a relevant busy period with size cutoff~$x$.
  However, in reality, the size cutoff decreases as $j$ receives service.
  We can use this to give a tighter bound on the amount of new work performed.

  Let $r_j$ be the amount of relevant work seen by $j$ on arrival.
  Note that $r_j$ is also the amount of old work that will be done while $j$
  is in the system.

  Starting from the time of $j$'s arrival, after at most $B_{\le x}(r_j)$ time,
  $j$ must enter service. During this busy period, an amount of work is performed
  equal to $r_j$ plus all arrivals during this busy period.

  More generally, for any amount of time $s \le x$,
  after at most a relevant busy period started by $r_j + ks$ work,
  $j$ must have received $s$ service.
  This holds because even if the servers finish
  all the old work and all the new work that has arrived so far,
  the servers must still complete $ks$ combined tagged and virtual work.
  Of this tagged and virtual work,
  at least $s$ must be tagged work, namely serving~$j$.
  This means that the first $dt$ service of $j$ must be completed by time
  \begin{equation*}
    B_{\le x}(r_j) + B_{\le x}(k \cdot dt).
  \end{equation*}
  The next $dt$ service of $j$ must be completed by time
  \begin{equation*}
    B_{\le x}(r_j) + B_{\le x}(k \cdot dt) + B_{\le x-dt}(k \cdot dt),
  \end{equation*}
  because the cutoff for entering the relevant busy period decreases as $j$ receives service.
  Similarly, the following $dt$ service of $j$ must be completed by time
  \begin{equation*}
    B_{\le x}(r_j) + B_{\le x}(k \cdot dt) + B_{\le x-dt}(k \cdot dt) + B_{\le x-2\,dt}(k \cdot dt).
  \end{equation*}
  This pattern continues as $j$ receives service.
  The descending size cutoff yields
  the same sort of relevant busy period as in the traditional tagged job analysis
  of SRPT\1/ \cite{schrage1966queue}.
  Recalling that $r_j$ is drawn from the distribution $\rwork_{\le x}^\srptk$
  yields the following bound on the mean response time of $j$:
  \begin{equation}
    \label{eq:improved-busy-bound}
    T^\srptk(x) \le B_{\le x}\bigl(\rwork_{\le x}^\srptk\bigr) + \int_0^x B_{\le t}(k\cdot dt).
  \end{equation}
  With \eqref{eq:improved-busy-bound}, we have improved upon~\eqref{eq:busy_bound}.

  Next, we will improve upon Lemma~\ref{lem:delta}. We consider a pair of systems
  experiencing the same arrival sequence: System~1, which uses PSJF\1/, and System~$k$,
  which uses SRPT\k/.

  Recall from Section~\ref{sec:psjf} that under PSJF\1/, a job~$\ell$ is relevant
  to~$j$ if~$\ell$
  has \emph{original} size at most $x$.
  In contrast, under SRPT\k/, a job~$\ell$ is
  relevant to~$j$ if~$\ell$ has \emph{remaining} size at most $x$.

  We define $\Delta'_{\le x}(t)$ to be the difference between the amounts of relevant work
  in the two systems at time~$t$.
  Using Lemma~\ref{lem:improved-delta} (proof deferred),
  we obtain a bound on $\Delta'_{\le x}(t)$
  tighter than the analogous bound in Lemma~\ref{lem:delta}.

\begin{lemma}
  \label{lem:improved-delta}
  The difference in relevant work between Systems 1 and $k$ is bounded by
  \begin{equation*}
    \Delta'_{\le x}(t) \le x\cdot\busy^\B_{\le x}(t)
  \end{equation*}
  where $\busy_{\le x}^\B(t)$ is the number of servers in System~$k$ which are busy with relevant work at time $t$.
\end{lemma}

\begin{proof}
  We define few-jobs intervals and many-jobs intervals as in Section~\ref{sec:srpt-rwork}.
  Note that $\busy_{\le x}^\B(t) = k$ during a many-jobs interval, and that
  $\busy_{\le x}^\B(t)$ is the number of jobs in the system during a few-jobs interval.

  The case where $t$ is in a few-jobs interval is simple: there are exactly $\busy_{\le x}^\B(t)$
  jobs in System $k$ at time $t$, each of remaining size at most $x$, so
  \begin{equation*}
    \Delta'_{\le x}(t) \le x\cdot\busy_{\le x}^\B(t).
  \end{equation*}

  Suppose instead that $t$ is in a many-jobs interval,
  in which case $\busy_{\le x}^\B(t) = k$. Let time $s$
  be the start of the many-jobs interval containing $t$.
  Over the interval $[s, t]$, the
  same amount of relevant work arrives in both systems,
  because relevant arrivals are the same under SRPT and PSJF.
  Upon arrival a job's original and remaining sizes are equal.
  The other two categories of relevant work over the interval follow
  the same arguments as in the proof of Lemma~\ref{lem:delta}. Thus,
  \begin{equation*}
    \Delta'_{\le x}(t) \le \Delta'_{\le x}(s).
  \end{equation*}

  It therefore suffices to show $\Delta'_{\le x}(s) \le kx$.
  As in Lemma~\ref{lem:delta}, a many-jobs interval can begin due to
  the arrival of a relevant job, or due an irrelevant job in System~$k$ becoming relevant.
  In the case of an arrival, the same arrival occurs in System~1,
  and must be relevant in System~1, so
  \begin{equation*}
    \Delta'_{\le x}(s) = \Delta_{\le x}(s^-) \le (k-1)x,
  \end{equation*}
  because $s^-$, the instant before~$s$, is in a few-jobs interval.
  In the case of an irrelevant job in System~$k$ becoming relevant, by the same argument
  as in the proof of Lemma~\ref{lem:delta},
  \begin{equation*}
    \Delta'_{\le x}(s) \le \rwork^\B_{\le x}(s) \le kx.
    \qedhere
  \end{equation*}
\end{proof}

  Continuing the proof of Theorem~\ref{thm:improved-response},
  we are now ready to prove the stronger bound.
  From \eqref{eq:improved-busy-bound}, we know
  \begin{equation*}
    T^\srptk(x) \le B_{\le x}\bigl(\rwork_{\le x}^\srptk\bigr) + \int_0^x B_{\le t}(k\cdot dt).
  \end{equation*}
  By plugging in Lemma~\ref{lem:vwork} and Lemma~\ref{lem:improved-delta},
  we find that
  \begin{align*}
    &T^\srptk(x) \\
    &\le B_{\le x}\bigl(\rwork_{\le x}^\psjf + x\cdot\busy^\srptk_{\le x}\bigr) + \int_0^x B_{\le t}(k\cdot dt)\\
    &= B_{\le x}\bigl(\rwork_{\le x}^\psjf\bigr) + B_{\le x}\bigl(x\cdot\busy^\srptk_{\le x}\bigr) + \int_0^x B_{\le t}(k\cdot dt)\\
    &= W^\psjf(x) + B_{\le x}\bigl(x\cdot\busy^\srptk_{\le x}\bigr) + \int_0^x B_{\le t}(k\cdot dt),
  \end{align*}
  where $\busy^\srptk_{\le x}$ is the steady state number of servers which are busy with relevant jobs under SRPT\k/.
  Taking expectations yields
  \begin{multline*}
    \E[\big]{T^\srptk(x)} \le \E[\big]{W^\psjf(x)}
    + \E[\big]{B_{\le x}\bigl(x\cdot\busy^\srptk_{\le x}\bigr)} \\
    + \int_0^x \E{B_{\le t}(k\cdot dt)}.
  \end{multline*}
  From the literature \cite{wierman2005nearly}, we know that
  \begin{equation*}
    \E[\big]{W^\psjf(x)} =
    \frac{\int_0^x \lambda t^2f_S(t) \,dt}{2(1-\rho_{\le x})^2}.
  \end{equation*}
  By the expectation of a relevant busy period, from Definition~\ref{def:busy_period},
  \begin{equation*}
    \int_0^x \E{B_{\le t}(k\cdot dt)} =  \int_0^x \frac{k}{1-\rho_{\le t}} \,dt.
  \end{equation*}
  Similarly,
  \begin{equation*}
    \E[\big]{B_{\le x}\bigl(x\cdot\busy^\srptk_{\le x}\bigr)} = \dfrac{\E[\big]{x\cdot\busy^\srptk_{\le x}}}{1-\rho_{\le x}}.
  \end{equation*}

  The average rate at which the SRPT\k/ system performs relevant work is
  $\E[\big]{\busy^\srptk_{\le x}}/k$, since each busy server does work at rate $1/k$.
  Because the system is stable, the rate at which
  relevant work is done must equal the rate at which relevant work enters the system,
  namely $\rho_{\le x}$.
  Thus, $\E[\big]{\busy^\srptk_{\le x}} = k \rho_{\le x}$, so
  \begin{equation*}
    \E[\big]{B_{\le x}\bigl(x\cdot\busy^\srptk_{\le x}\bigr)} = \dfrac{k \rho_{\le x}x}{1-\rho_{\le x}},
  \end{equation*}
  yielding the desired bound.
\end{proof}

\section{Matuszewska Index}
\label{app:matuszewska}

The heavy-traffic results in this paper, such as Theorem~\ref{thm:ratio},
assume that
the service requirement distribution~$S$ is not too heavy-tailed.
Specifically, we require that either $S$ is bounded or that
the \emph{upper Matuszewska index} of the tail of $S$ is less than~$-2$.
This is slightly stronger than assuming that $S$ has finite variance.
The formal definition of the upper Matuszewska index is the following.

\begin{definition}
  Let $f$ be a positive real function.
  The \emph{upper Matuszewska index} of~$f$, written $M(f)$,
  is the infimum over $\alpha$ such that
  there exists a constant $C$ such that for all $\gamma > 1$,
  \begin{equation*}
    \lim_{x \to \infty}\frac{f(\gamma x)}{f(x)} \le C\gamma^\alpha.
  \end{equation*}
  Moreover, for all $\Gamma > 1$,
  the convergence as $x \to \infty$ above must be uniform in $\gamma \in [1, \Gamma]$.
\end{definition}

The condition $M(\bar F_S) < -2$, where $\bar F_S$ is the tail of~$S$,
is intuitively close to saying that $F_S(x) \le Cx^{-2 - \epsilon}$
for some constant $C$ and some $\epsilon > 0$.
Roughly speaking, this means that
$S$ has a lighter tail than a Pareto distribution with $\alpha = 2$.

\section{SRPT\1/ in Heavy Traffic}
\label{app:srpt-1-bound}

\begin{replemma}{lem:srpt-1-bound}
  In an \mg1/ with any service requirement distribution~$S$ which is either
  \begin{enumerate*}[(i)]
  \item
    \label{item:bounded}
    bounded or
  \item
    \label{item:unbounded}
    unbounded with tail function of upper Matuszewska index\footnote{%
      See Section~\ref{sec:matuszewska} or Appendix~\ref{app:matuszewska}.} less than~$-2$,
  \end{enumerate*}
  \[\lim_{\rho \rightarrow 1} \frac{\log\Bigl(\frac{1}{1-\rho}\Bigr)}{\E*{T^\srpt}} = 0.\]
\end{replemma}

\begin{proof}
  \Citet{Lin} show in their Theorem~1 that if $S$ is bounded, then
  \begin{equation*}
    \E[\big]{T^\srpt} = \Theta\biggl(\frac{1}{1-\rho}\biggr),
  \end{equation*}
  proving case~\ref{item:bounded}.
  They also show in their Theorem~2 that if the upper Matuszewska index of the tail of $S$ is less than~$-2$, then
  \begin{equation*}
    \E[\big]{T^\srpt} = \Theta\biggl(\frac{1}{(1-\rho)G^{-1}(\rho)}\biggr),
  \end{equation*}
  where $G^{-1}(\cdot)$ is the inverse of $G(x) = \rho_{\le x}/\rho$. In their proof of Theorem~2, they also show that
  \begin{equation*}
    \lim_{\rho \rightarrow 1} \log \left(\frac{1}{1-\rho} \right) \cdot (1-\rho)G^{-1}(\rho) = 0,
  \end{equation*}
  proving case~\ref{item:unbounded}.
\end{proof}

\section{FB\1/ in Heavy Traffic}
\label{app:fb-1-bound}

\begin{replemma}{lem:fb-1-bound}
  In an \mg1/ with any service requirement distribution~$S$ which
  is unbounded with tail function of upper Matuszewska index less than~$-2$,
  \begin{equation*}
    \lim_{\rho \rightarrow 1} \frac{\E{R(S)}}{\E[\big]{T^\fb}} = 0.
  \end{equation*}
\end{replemma}

\begin{proof}
  Recall that
  \begin{align*}
    W(x) &= \E{B_{\bar x}(\rwork^\fbk_{\bar x})} \\
    R(x) &= \E{B_{\bar x}(x)}.
  \end{align*}
  The standard analysis of FB\1/
  \cite{schrage1967queue, harchol2013performance}
  shows
  \begin{equation*}
    \E[\big]{T^\fb} = \E{W(S)} + \E{R(S)}.
  \end{equation*}
  \Citet[Equation~(4.3)]{kamphorst2017heavy}
  decompose $\E[\big]{T^\fb}$ into a sum of three functions of the load~$\rho$,
  \begin{equation*}
    \E[\big]{T^\fb} = X(\rho) + Y(\rho) + Z(\rho),
  \end{equation*}
  such that
  \begin{align*}
    \E{W(S)} &= Z(\rho) + \frac{1}{2}Y(\rho) \\
    \E{R(S)} &= X(\rho) + \frac{1}{2}Y(\rho).
  \end{align*}
  \Citet[Section~4.1.1]{kamphorst2017heavy} then show that
  \begin{equation*}
    \lim_{\rho \to 1} \frac{X(\rho)}{Z(\rho)}
    = \lim_{\rho \to 1} \frac{Y(\rho)}{Z(\rho)} = 0,
  \end{equation*}
  which implies the desired limit
  \begin{equation*}
    \lim_{\rho \rightarrow 1} \frac{\E{R(S)}}{\E[\big]{T^\fb}}
    = \lim_{\rho \rightarrow 1}
      \frac{X(\rho) + \frac{1}{2}Y(\rho)}{X(\rho) + Y(\rho) + Z(\rho)}
    = 0.
    \qedhere
  \end{equation*}
\end{proof}

\section{Flawed Interchange Arguments}
\label{app:interchange}

\Citet[Theorem~2.1]{down2006multi} claim that SRPT\k/ is optimal in the sense of minimizing the completion time of the $n$th job for all $n$, under the additional assumption that all servers are busy at all times.
Unfortunately, this claim is false.
The proof attempts to use an interchange argument,
mimicking the classic proof of the optimality of SRPT\1/ \cite{schrage1968letter}.
However, the specified interchange can result in
the same job running on two servers simultaneously,
which is of course not possible.

A concrete counterexample is the following: let $k=2$, and let jobs of size $1, 1, 2$ and $2$ arrive at time $0$. Recall that a job of size $x$ must be in service for $kx$ time to complete. SRPT\k/ completes its third job at time $6$, while a policy which serves a job of size $2$ over the interval $[0, 4)$ and jobs of size $1$ over the intervals $[0, 2)$ and $[2, 4)$ would finish its third job at time~$4$.
Moreover, more complicated counterexamples exist which show that multiserver SRPT does not minimize mean response time even if all servers are busy at all times.

A similar error occurs in a claim by \citet[Theorem~2.1]{wu2004scheduling}
that FB\k/ is optimal among policies that do not have access to job size information
when the service requirement distribution has decreasing hazard rate.
Again the proof given is an interchange argument,
and again the specified interchange can result in
the same job running on two servers simultaneously.

\end{document}

%% file: fig_server_diagram.tex
\begin{tikzpicture}[semithick, scale=0.3]
  \begin{scope}
    \node at (7.5, 5.25) {\textsc{Single-Server System}};
    \draw (10.5, 0) circle (1.5);
    \node[below] at (10.5, -1.4) {speed $1$};
    \draw[->] (12, 0) -- ++(0.75, 0);
    \draw (5, -1.25) -- (9, -1.25) -- (9, 1.25) -- (5, 1.25);
    \draw (7.5, -1.25) -- ++(0, 2.5);
    \draw (6, -1.25) -- ++(0, 2.5);
    \draw[->] (4.25, 0) node [left] {$\lambda$} -- ++(0.75, 0);
  \end{scope}
  \begin{scope}[shift={(12, 0)}]
    \node at (9, 5.25) {\textsc{$k$-Server System}};
    \newcommand{\server}[2]{%
      \draw (9, #1) -- (12, #2) ++(0.5, 0) circle (0.5)
        ++(0.7, -0.35) node [below] {speed $1/k$} ++(-0.2, 0.35) edge[->] ++(0.5, 0);}
    \server{0.98}{3.5}
    \server{0.21}{0.75}
    \filldraw (12.5, -1.4) circle (1.5mu) ++(0, -0.5) circle (1.5mu) ++(0, -0.5) circle (1.5mu);
    \server{-0.98}{-3.5}
    \draw (5, -1.25) -- (9, -1.25) -- (9, 1.25) -- (5, 1.25);
    \draw (7.5, -1.25) -- ++(0, 2.5);
    \draw (6, -1.25) -- ++(0, 2.5);
    \draw[->] (4.25, 0) node [left] {$\lambda$} -- ++(0.75, 0);
  \end{scope}
\end{tikzpicture}


%% file: fig_waiting_residence.tex
\begin{tikzpicture}[scale=0.3]
  \fill[green!65!black!30] (11, -0.5) rectangle (13, 0.5);
  \fill[green!65!black!30] (14, -0.5) rectangle (15, 0.5);
  \fill[green!65!black!30] (16.75, -0.5) rectangle (18, 0.5);
  \draw[->] (14.375, -2.5) -- (12.75, -0.5);
  \draw[->] (15.25, -2.5) -- (14.5, -0.5);
  \draw[->] (16.625, -2.5) -- (17.25, -0.5);
  \node at (15.325, -3.2) {$j$ in service};
  \draw[thick] (0, 0) node [label={left:$j$ arrives}] {} -- (18, 0) node [label={right:$j$ departs}] {};
  \draw[thick] (0, -0.5) -- ++(0, 1);
  \draw[thick] (11, -0.5) -- ++(0, 1);
  \draw[thick] (18, -0.5) -- ++(0, 1);
  \draw[decorate, decoration={brace, amplitude=5pt}] (0.15, 1) -- (10.85, 1);
  \node[align=center] at (5.5, 3) {waiting time\\ $W^\srpt(x)$};
  \draw[decorate, decoration={brace, amplitude=5pt}] (11.15, 1) -- (17.85, 1);
  \node[align=center] at (14.5, 3) {residence time\\ $R^\srpt(x)$};
  \draw[decorate, decoration={brace, amplitude=5pt}] (17.85, -1) -- (0.15, -1);
  \node[align=center] at (9, -3) {response time\\ $T^\srpt(x)$};
\end{tikzpicture}


%% file: fig_delta.tex
\begin{tikzpicture}[scale=0.3]
  \fill[red!25] (3.5, 0) rectangle (6, 11);
  \fill[red!25] (9.5, 0) rectangle (10.5, 11);
  \fill[red!25] (13, 0) rectangle (14.5, 11);
  \fill[red!25] (17, 0) rectangle (19, 11);
  \draw[->] (8.5, 13) -- (5.5, 11);
  \draw[->] (10.75, 13) -- (10, 11);
  \draw[->] (12.625, 13) -- (13.75, 11);
  \draw[->] (14.5, 13) -- (17.5, 11);
  \node at (11.5, 13.8) {System~$1$ is empty};
  \draw[semithick, ->] (-0.5, 0) node [left] {$0$} -- (21, 0)
    node [right] {time~$t$};
  \draw[semithick, ->] (0, -0.5) node [below] {$0$} -- (0, 11)
    node [above, align=center] {work difference\\$\Delta_{\le x}(t)$};
  \draw[semithick] (-0.5, 9) node [left] {$kx$} -- (0, 9);
  \draw[semithick, densely dotted] (0, 9) -- (21, 9);
  \draw[decorate, decoration={brace, amplitude=5pt}] (19.85, -0.7) -- (0.15, -0.7);
  \node[align=center] at (10, -3) {many-jobs interval\\ all servers in System~$k$ occupied};
  \draw[semithick] (20, -0.5) -- (20, 0);
  \draw[semithick, densely dotted] (20, 0) -- (20, 11);
  \begin{scope}[shift={(0, -1)}]
    \draw[thick] (0, 9) -- (3.5, 9) -- (6, 6.5) -- (9.5, 6.5) -- (10.5, 5.5) --
      (13, 5.5) -- (14.5, 4) -- (17, 4) -- (19, 2) -- (20, 2);
  \end{scope}
\end{tikzpicture}
